\documentclass[lettersize,journal]{IEEEtran}
\usepackage{etoolbox} 
\newbool{hide_del}
\newbool{replace}

\IEEEoverridecommandlockouts
 
\usepackage{textcomp} 
\usepackage[OT2,T1]{fontenc}
\usepackage{xspace}
\usepackage[normalem]{ulem} 
\usepackage{cite}
\usepackage{enumitem} 

\usepackage{pgfplots}
\usepackage[footnotesize]{caption}
\usepackage{subfig}
\usepackage{graphicx}
\usepackage{tabularx}
\usepackage{booktabs} 
\usepackage{circuitikz}
\usetikzlibrary{external}
\usepackage{amsmath, amsbsy, amssymb, amsthm, amsfonts,bm}

\allowdisplaybreaks
\usepackage{mathtools} 
\usepackage{mathrsfs}
\usepackage{multirow}
\usepackage{gensymb} 
\usepackage{cancel} 
\usepackage{ifthen}
\usepackage{stmaryrd} 
\usepackage{dsfont} 

\newtheorem{theorem}{Theorem}
\newtheorem{lemma}[theorem]{Lemma}
\newtheorem{proposition}[theorem]{Proposition}
\newtheorem{corollary}{Corollary}

\newtheorem{assumption}{Assumption}

\theoremstyle{definition}

\newtheorem{remark}{Remark}

\newlist{todolist}{itemize}{2}
\setlist[todolist]{label=$\square$}
\usepackage{pifont}
\usepackage[top=1.5cm,bottom=2.0cm,left=1.2cm,right=1.2cm]{geometry}
\usepackage{fancyhdr}
\usepackage{setspace}
\usepackage{lscape}
\usepackage{pdfpages} 
\usepackage{xcolor}
\pgfplotsset{compat=newest}
\pgfplotsset{plot coordinates/math parser=false}
\newlength\figureheight
\newlength\figurewidth
\pgfplotsset{compat=newest}
\pgfplotsset{plot coordinates/math parser=false}

\usepackage[hyphens]{url}
\usepackage[hidelinks]{hyperref}
	
\newcommand{\splev}{{K_0}}

\newcommand{\bb}{\mathbb}
\newcommand{\bs}{\boldsymbol}
\newcommand{\bc}[1]{\boldsymbol{\mathcal #1}}
\newcommand{\cl}{\mathcal}

\newcommand{\wt}{\widetilde}

\newcommand{\st}{%
  \ifmmode
  \ {\rm s.t.}\ %
  \else%
  \emph{s.t.}\@\xspace%
  \fi%
}
\newcommand{\whp}{\emph{w.h.p.}\xspace}
\newcommand{\ie}{\emph{i.e.}, }
\newcommand{\eg}{\emph{e.g.}, }

\newcommand{\aka}{\emph{a.k.a.} }
\newcommand{\rmc}{\mathrm{c}}
\newcommand{\rmh}{\mathrm{h}}

\newcommand{\Rbb}{\bb{R}} 
\newcommand{\Cbb}{\bb{C}} 

\DeclareMathOperator{\vect}{vec} 
\DeclareMathOperator{\Id}{{\rm \bf I}} 
\newcommand{\tinv}[1]{{\textstyle\frac{1}{#1}}}

\newcommand{\im}{\mathrm{i}\mkern1mu} 
\newcommand{\conj}[1]{#1^{*}} 

\newcommand{\proba}[2][]{
  {
    \ifx&#1& \bb{P}\left[#2\right] 
    \else\bb{P} #1[#2 #1]\fi
  }
}

\newcommand{\E}[2][]{\ifthenelse{ \equal{#2}{} }{\bb{E}_{#1}}{\bb{E}_{#1}\left[#2\right]}} 
\newcommand{\iid}{%
  \ifmmode
  \mathrm{i.i.d.}%
  \else%
  i.i.d.\@\xspace%
  \fi%
}
\newcommand{\simiid}{\sim_{\iid}}

\DeclareSymbolFont{cyrletters}{OT2}{wncyr}{m}{n}
\DeclareMathSymbol{\Sha}{\mathalpha}{cyrletters}{"58}
\newcommand{\ud}{\mathrm{d}} 
\newcommand{\supp}{{\rm supp}\,} 
\DeclareMathOperator{\tr}{tr} 
\DeclareMathOperator{\diag}{diag}
\newcommand{\dmat}[1]{#1_\text{d}} 
\newcommand{\transpose}[1]{#1^\top} 
\newcommand{\norm}[3][]{#1\lVert#2#1\rVert_{#3}} 
\newcommand{\scp}[3][]{#1\langle #2, #3 #1\rangle} 
\newcommand{\fro}[3][]{#1\langle #2, #3 #1\rangle_{\text{F}}} 

\DeclareMathOperator*{\argmin}{arg\,min}

\newcommand{\upto}[1]{\llbracket #1 \rrbracket} 
\newcommand{\RIP}{\text{RIP}_{\ell_2/\ell_2}}
\newcommand{\RIPto}{\text{RIP}_{\ell_2/\ell_1}}

\newcommand{\intMa}[1][]{
   \ifthenelse{ \equal{#1}{} }{\bs{\cl I}}{\cl I_{#1}}
}
\newcommand{\intMOm}[2][]{
   \ifthenelse{ \equal{#2}{} }{\bs{\cl I}_\Omega}{\bs{\cl I}_\Omega #1[ #2 #1]} 
}
\newcommand{\intMOmT}[2][]{
   \ifthenelse{ \equal{#2}{} }{\bs{\cl I}_\Omega^*}{\bs{\cl I}_\Omega^* #1[ #2 #1]} 
}
\newcommand{\intMOmb}[2][]{
   \ifthenelse{ \equal{#2}{} }{\bs{\cl I}_{\Omega_b}}{\bs{\cl I}_{\Omega_b} #1[ #2 #1]} 
}
\newcommand{\intMOmt}[2][]{
   \ifthenelse{ \equal{#2}{} }{\bs{\cl I}_{\Omega(t)}}{\bs{\cl I}_{\Omega(t)} #1[ #2 #1]} 
}
\newcommand{\dintMOm}[2][]{
   \ifthenelse{ \equal{#2}{} }{\tilde{\bs{\cl I}}_\Omega}{\tilde{\bs{\cl I}}_\Omega #1[ #2 #1]} 
}
\newcommand{\intCirc}[1][]{
   \ifthenelse{ \equal{#1}{} }{\bs{\cl J}}{\cl J_{#1}} 
}
\newcommand{\dintCirc}[1][]{
   \ifthenelse{ \equal{#1}{} }{\tilde{\bs{\cl J}}}{\tilde{\cl J}_{#1}} 
}
\newcommand{\ropA}{\bs{\cl A}} 

\newcommand{\ts}{\textstyle}

\newcommand{\skyvar}{\theta_{\rm c}} 
\newcommand{\intmap}{\theta_{\rm c}} 
\newcommand{\vintmap}{\theta_{\rm v}} 
\newcommand{\intdisc}{\bs\theta} 
\newcommand{\intdiscel}{\theta} 

\newcommand{\btoB}{_{b=1}^B}
\newcommand{\CRIop}{\bs{RG}_0\bs F}
\newcommand{\acqop}{\bs\Psi}
\newcommand{\acqopn}{\wt{\bs\Psi}}
\newcommand{\imop}{\bs\Phi}
\newcommand{\imopn}{\wt{\bs\Phi}}
\newcommand{\covmat}[1][]{
   \ifthenelse{ \equal{#1}{} }{\bs C}{C_{#1}} 
}
\newcommand{\sampcovmat}[1][]{
   \ifthenelse{ \equal{#1}{} }{\bs C}{C_{#1}} 
}

\renewcommand{\leq}{\leqslant}
\renewcommand{\geq}{\geqslant}

\makeatletter
\newcommand{\mathleft}{\@fleqntrue\@mathmargin0pt}
\newcommand{\mathcenter}{\@fleqnfalse}

\fancypagestyle{firstpage}{%
    
    \cfoot{\thepage}
}

\fancypagestyle{others}{%
    \cfoot{\thepage}
    
}

\usepackage{hyperref}
\hypersetup{
    colorlinks=true,
    linkcolor=black,
    filecolor=magenta,      
    urlcolor=cyan,
    citecolor = magenta
}
\usepackage{cleveref}
\crefformat{equation}{(#2\color{blue}#1#3)} 

\definecolor{NavyBlue}{RGB}{47,113,194}
\definecolor{greenmat}{RGB}{0,136,43}
\definecolor{org}{RGB}{222,106,16}
\definecolor{greenmat2}{RGB}{0,166,43}
\definecolor{org2}{RGB}{240,100,16}
\definecolor{MH}{RGB}{160,160,0}
\definecolor{SS}{RGB}{200,60,120}

\newcommand{\replace}[2][\@empty]
{ \ifbool{replace}
  {
  #2
  }
  {
  \ifx\@empty#1\relax \textcolor{greenmat2}{[BY]: #2} 
  \else\textcolor{org2}{[REPLACE]: #1} \\ \textcolor{greenmat2}{[BY]: #2}\fi
  }
}
\newcommand{\br}[1]{\textcolor{red}{[\textbf{BR:#1}]}}
\renewcommand{\br}[1]{}

\IEEEoverridecommandlockouts\IEEEpubid{\makebox[\columnwidth]{ 978-1-6654-4442-2/21 /\$31.00~\copyright~2021 IEEE \hfill} \hspace{\columnsep} \makebox[\columnwidth]{ }}

\title{Compressive radio-interferometric sensing with random beamforming as rank-one signal covariance projections}

\author{Olivier Leblanc\IEEEauthorrefmark{1}, Yves Wiaux\IEEEauthorrefmark{2}, and Laurent Jacques\IEEEauthorrefmark{1}
\thanks{\IEEEauthorrefmark{1} E-mail: {\em \{o.leblanc, 
laurent.jacques\}@uclouvain.be}. ISPGroup, INMA/ICTEAM, UCLouvain, 
Louvain-la-Neuve, Belgium. OL is funded by Belgian National 
Science Foundation (F.R.S.-FNRS). Part of this research is funded by the Fonds de la Recherche Scientifique -- FNRS under Grants n$^\circ$ T.0136.20 
(Learn2Sense) and T.0160.24 (QuadSense).} 
\thanks{\IEEEauthorrefmark{2}Heriot-Watt University, Edinburgh, UK.}}

\begin{document}


\maketitle

\pagestyle{others}
\thispagestyle{firstpage}

\begin{abstract}
Radio-interferometry (RI) observes the sky at unprecedented angular resolutions, enabling the study of several far-away galactic objects such as galaxies and black holes. In RI, an array of antennas probes cosmic signals coming from the observed region of the sky. The covariance matrix of the vector gathering all these antenna measurements offers, by leveraging the Van Cittert-Zernike theorem, an incomplete and noisy Fourier sensing of the image of interest. The number of noisy Fourier measurements---or \emph{visibilities}---scales as $\cl O(Q^2B)$ for $Q$ antennas and $B$ short-time integration (STI) intervals. We address the challenges posed by this vast volume of data, which is anticipated to increase significantly with the advent of large antenna arrays, by proposing a compressive sensing technique applied directly at the level of the antenna measurements. First, this paper shows that \emph{beamforming}---a common technique of dephasing antenna signals---usually used to focus some region of the sky, is equivalent to sensing a rank-one projection (ROP) of the signal covariance matrix. We build upon our recent work~\cite{leblanc24} to propose a compressive sensing scheme relying on random beamforming, trading the $Q^2$-dependence of the data size for a smaller number $P$ of ROPs. We provide image recovery guarantees for sparse image reconstruction. Secondly, the data size is made independent of $B$ by applying $M$ random modulations of the ROP vectors obtained for the STI. The resulting sample complexities, theoretically derived in a simpler case without modulations and numerically obtained in phase transition diagrams, are shown to scale as $\cl O(K)$ where $K$ is the image sparsity. This illustrates the potential of the approach. 
\end{abstract}


\section{Introduction}
\label{sec:intro}

Radio-interferometry allows observations of the sky with fine angular resolutions and sensitivities in radio astronomy. The covariance of signals acquired at antennas placed strategically around the Earth yields a noisy and nonuniformly partial Fourier sampling of the image of interest, \ie a small portion of the celestial sphere accessible to the instrument. RI is fundamental in studying a wide range of astronomical phenomena, including star formation, galaxy evolution, and the cosmic microwave background.

A tremendous amount of data is processed and collected every day. This includes the Fourier samples contained in all the covariance matrices, which are consecutive in time due to the rotation of the Earth. For the low-frequency array (LOFAR), this represents around 5 petabytes ($5~10^{15}$ bytes) of data per year~\cite{van_haarlem13}. 

Compression is therefore becoming increasingly essential for reducing data size and ensuring the scalability of RI, particularly with upcoming arrays like the \emph{Square Kilometer Array}~\cite{braun19ska}, which will involve thousands of antennas. The issue faced by \textit{post-sensing} compression techniques is that they require computing the signal covariance matrix to compress it afterward, hence temporarily storing the uncompressed data. Furthermore, the compression mechanism often impedes the forward model calculation, repeatedly called in iterative reconstruction algorithms.

In this work, we present a \emph{compressive radio-interferometric} (CRI) sensing scheme with \emph{(i)} a cheap acquisition computation, \emph{(ii)} a minimal sample complexity, and \emph{(iii)} supported by an established theoretical background. Namely, we highlight the correspondence between beamforming and ROP applied to the covariance matrix of the antenna measurement vector in order to compress the visibilities directly at the level of the antennas without going through the computation of this covariance matrix. The proposed two-layer compression scheme first consists of applying ROP to the covariance matrices of each short-time integration (STI) interval. Then, the concatenated ROP vectors are further compressed by applying random modulations across the STI intervals.  

This paper highlights the novel aspects of the acquisition process by introducing a compressive sensing operator $\acqopn$ and subsequently develops the corresponding forward imaging model $\imopn$, along with the associated recovery guarantees. A more detailed exploration of the computational complexity of the forward model, as well as image reconstruction methodologies, is reserved for a parallel submission.

\subsection{Related Work} \label{sec:related}

There exist several compression techniques in the literature, either \textit{post-sensing} compression\footnote{A post-sensing compression technique necessitates computing the uncompressed data before applying compression.} or \textit{compressive sensing}, that have a connection with our method.

Among the post-sensing compression techniques for RI, compressing the observation vector that collects all the visibilities by multiplication with a Gaussian random matrix (with much fewer rows than columns) was discussed in~\cite{kartik17} and shown to significantly increase the computational cost of the forward model.~\cite{kartik17} proposed an efficient Fourier reduction model approximating the optimal (in the least-squares sense) dimensionality reduction that projects the data with the adjoint of $\bs U$ which contains the left singular vectors of the SVD of the visibility operator $\bs\Phi:=\bs U \bs\Sigma\bs V^*$. On another hand, \emph{baseline-dependent averaging}~\cite{wijnholds18,atemkeng18} offers an effective reduction in data volume by averaging consecutive visibilities corresponding to short baselines, \ie those with nearly constant frequency locations. Finally, as most of the reconstruction algorithms involve the matrix-to-vector multiplication with the matrices $\bs\Phi$ and $\bs\Phi^*\bs\Phi$, storing the \emph{dirty map}---that is the mere application of the adjoint sensing operation to the observed visibilities---instead of the visibility vector was considered as a practical compression technique~\cite{kartik17}. It can even be computed \textit{on-the-fly} during the acquisition~\cite{thyagarajan17,krishnan23}. \\

Our scheme is a \emph{compressive sensing} technique relying on \emph{random beamforming}, which had already been highlighted in~\cite{aecal15}. Our decomposition of the beamforming capabilities into direction-dependent gains per antenna, then direction-agnostic gains per antenna for the projection of the measurement vector, corresponds to the random beamforming strategy R3 described in~\cite{aecal15}. The novelty in this paper is the emphasis on the theory of ROP, allowing us to derive the sample complexity in a simplified case. In the last decade, different works have provided matrix recovery guarantees via ROPs. In particular,~\cite{cai15} derived sample complexities for the recovery of low-rank matrices, showing that a matrix of rank $r \ll Q$ could be recovered from a number $n\gtrsim rQ^2$ of random ROPs, and~\cite{chen2015exact} studied the reconstruction of a signal covariance matrix from symmetric ROPs when the matrix satisfies one structural assumption among low-rankness, Toeplitz, sparsity, and joint sparsity and rank-one. Furthermore, we tackled the challenge regarding time-dependence by proposing \emph{random modulations} of the different rank-one projected vectors. Using random modulations for compressive imaging purposes is not new, and was considered, for instance, in coded aperture~\cite{fenimore78,wagadarikar08}, lensless imaging, and in binary mask schemes~\cite{koller15}. It is used for computer vision and robotics but also for astronomical and medical imaging applications. \\

In the context of compressive sensing, our CRI model is analogous to the one of \emph{Multicore Fiber Lensless Imaging} (MCFLI)~\cite{leblanc24}. In a nutshell; in both cases, the sensing model amounts to sampling the frequency content of the object of interest over a subset of frequencies related to the spatial configuration of fundamental components, \ie the antenna or core positions in CRI and MCFLI, respectively (this analogy is further developed in App.~\ref{app:MCFLI}).

\subsection{Contributions} \label{sec:contributions}

We provide several contributions to the modeling, understanding, and efficiency of RI when combined with random beamforming. 

\paragraph{Random beamforming as a rank-one projection model} Random beamforming of antenna signals before computing their correlation is shown to be equivalent to applying random ROPs of the initial signal covariance matrix. Replacing the $Q^2$ coefficients of this matrix with a small---but sufficient---number $P$ of these projections for each STI represents a compression technique that can be applied on-the-fly during the acquisition of the antenna measurements. Moreover, as a direct consequence of the \emph{Van Cittert-Zernike} theorem~\cite{vancittert34,zernike38}, we show that in expectation with respect to the variability of the sky distribution (or assuming very long integration time), the ROPs of the covariance matrix equal the ROPs of an interferometric matrix gathering the visibilities of the sky distribution map.

\paragraph{Structured rank-one projections via random modulations} We propose to further compress the ROP vectors corresponding to consecutive STI intervals, or \emph{batches}, by progressively integrating them after applying random modulations; this allows keeping a fixed final data size $PM$ during the acquisition for $P$ ROPs per batch and $M$ modulations, reducing the data size compared to $PB$ ROP elements. 

\paragraph{Model analysis and recovery guarantees} Formal guarantees on the recovery of the sky intensity distribution are provided in Sec.~\ref{sec:recovery} for a sensing scheme---simpler to analyze theoretically---that sums (without modulations) the ROP vectors across batches, namely \emph{integrated ROP}. 
Nonetheless, integrated ROP (IROP) and modulated ROP (MROP) are shown to be both equivalent to ROPs applied to a total covariance matrix gathering all matrices coming from different batches in blocks along its diagonal. Correspondingly, in expectation with respect to the sky variability (or over long integration time), IROP and MROP are matching a ROP scheme applied to a global interferometric matrix gathering the visibilities over all STI intervals. To our knowledge, this is the first work showing such an equivalence and exploiting it in the context of RI. Guarantees adapted to MROPs are expected to exist but not proven here. 
From a set of simplifying assumptions, we show that one can with high probability (\whp) robustly estimate a $K$-sparse image provided that the number of ROPs $P$ and the Fourier coverage with $V := Q(Q-1)B$ visibilities guided by both the number of antennas $Q$ and the number of batches $B$ are large compared to $\cl O(K)$ (up to logarithmic factors). Our analysis relies on showing that, \whp, the sensing operator $\bs{RGF}$ satisfies the $\ell_2/\ell_1$ restricted isometry property which enables us to estimate a sparse image with the \ref{eq:BPDN} program. This theoretical analysis is accompanied by phase transition diagrams obtained from extensive Monte Carlo numerical experiments using the MROP model.

\subsection{Notations and Conventions} \label{sec:notation}

Light symbols denote scalars (or scalar functions), and bold symbols refer to vectors and matrices (\eg $\eta \in \Rbb$, $g \in L_2(\Rbb)$, $\bs f \in \Rbb^N$, $\bs G \in \Cbb^{N 
\times N}$). We write $\im=\sqrt{-1}$; the identity operator (or $n \times n$ matrix) is $\Id$ (resp. $\Id_n$); the set of $Q \times Q$ Hermitian matrices in $\bb C^{Q \times Q}$ is denoted by $\cl H^Q$; the set of index components is $\upto M := \{1,\ldots,M\}$; $\{s_q\}_{q=1}^Q$ is the set $\{s_1, \ldots, s_Q\}$, and $(a_q)_{q=1}^Q$ the vector $(a_1,\ldots,a_Q)^\top$. The notations $\transpose{\cdot}$, 
$\conj{\cdot}$, 
$\tr$,  
$\times$,
$\scp{\cdot}{\cdot}$, 
$\otimes$,
correspond to the transpose, conjugate transpose, trace, cross product, inner product, and Kronecker product. 
The $q$-norm (or $\ell_q$-norm) is $\norm{\bs x}{q}:=  (\sum_{i=1}^n |x_i|^q)^{1/q}$ for $\bs x \in \bb C^n$ and $q\geq 1$, with ${\|\cdot\|}={\|\cdot\|_2}$, and $\norm{\bs x}{\infty} := \max_i |x_i|$. Given $\bs A \in \bb C^{n\times n}$, $\bs a \in \bb C^n$ and $\cl S \subset \upto{N}$, the matrix $\bs A_{\cl S}$ is made of the columns of $\bs A$ indexed in $\cl S$, the operator $\diag(\bs A) \in \bb C^n$ extracts the diagonal of $\bs A$,  $\bs D_{\bs a} := \diag(\bs a) \in \bb C^{n\times n}$ is the diagonal matrix such that $\diag(\bs a)_{ii} = a_i$, $\dmat{\bs A}=\diag(\diag(\bs A))$ zeros out all off-diagonal entries of $\bs A$, the \emph{hollow} version of $\bs A$ is $\bs A_\rmh := \bs A-\dmat{\bs A}$, and $\|\bs A\|$ is the operator norm of $\bs A$.  The direct and inverse 2-D continuous Fourier transforms are defined by $\hat g(\bs \chi) := \cl F[g](\bs \chi) := \int_{\bb R^2} g(\bs s) e^{-\im 2\pi \bs \chi^\top \bs s} \ud \bs s$, with $g: \bb R^2 \to \bb C^2$, $\bs\chi \in \bb R^2$, and $g[\bs s] = \cl F^{-1}[\hat g](\bs s) = 
\int_{\bb R^2} \hat{g}(\bs \chi) e^{ \im 2\pi \bs \chi^\top \bs s} \ud \bs \chi$. 
The notation $X_i \simiid \cl P$ indicates that the random variables $\{ X_i \}_{i=1}^N$ are \emph{independent and identically distributed} according to the distribution $\cl P$. The uniform distribution on a set $\cl A$ is denoted by $\cl U(\cl A)$. The expectation with respect to the random quantity $Q$ (scalar or vector) is written $\bb E_{Q}$. An estimate of a quantity of interest $q$ or $\bs q$ is denoted $\wt{q}$ or $\wt{\bs q}$, respectively.

Finally, anticipating on their complete explanation, we list in Table.~\ref{tab:list-key-notations} several key quantities used in this work as a reading guide.

\begin{table}[t]
\centering
\caption{Key quantities and dimensions used in this work.}
\scalebox{.7}{\begin{tabular}{|c|l|}
\hline
Key quantity&Summarized meaning\\
\hline
$Q$, $B$&Number of antennas, and short-time intervals (STI) or batches\\
$I$, $T$&Number of samples in each STI, and sampling time period\\
$V=Q(Q-1)B$&Number of visibilities for all batches\\
$P$, $M$&Number of ROPs (in IROP \& MROP) and modulations (MROP)\\
$s(\bs l,t)$&Cosmic signals in direction cosine $\bs l = (l,m)$\\
$\intmap(\bs l)$, $\vintmap(\bs l)$, &Complete, vignetted, \\
$\intdisc \in \bb R^N$&\ and  discretized sky distribution on a $N$-pixel grid\\
$\bs p_q(t)$, $x_q(t)$&Location and signal of the $q$-th antenna ($q \in \upto{Q}$)\\
$\bs x(t), \bs x_b[i] \in \bb C^Q$&Time-continuous and discretized antenna signals\\
$\bs C(t), \bs C_b$&Time-continuous and discretized sample covariance matrices\\
$\Omega(t)$, $\Omega_b$&Time-continuous and discretized antenna arrangements\\
$\cl V(t)$, $\cl V_b$&Time-continuous and discretized baseline sets\\
$\intMOmt{}$, $\intMOmb{}$&Time-continuous and discretized interferometric matrices\\
\hline
\end{tabular}}
\label{tab:list-key-notations}
\end{table}

\section{Acquisition and Sensing Model} \label{sec:model}

Here, we present three models compatible with radio interferometry. The first is the classical scheme computing frequency samples of the images, or \emph{visibilities}, from the signal covariance matrix. Then follow random Gaussian compression and baseline-dependent averaging; two post-sensing compression techniques acting on the visibilities. Finally, we propose a new compressive sensing scheme, coined MROP, occurring at the antenna level and circumventing the limitations raised in the other models.

\subsection{Classical Acquisition:\newline\phantom{A.\ }From the Antenna Signals to the Visibilities} \label{sec:meas}

This section provides a recapitulation of the classical sensing scheme. The radio-interferometric measurements and the link of their covariance matrix to the visibilities are derived. Then we show how the visibilities are usually accumulated over $B$ batches in order to obtain a sufficiently dense Fourier sampling of the image of interest.  

Let us consider a context of radio-astronomical imaging, as depicted in Fig.~\ref{fig:21_schematic}, with $Q$ antennas\footnote{We will write ``antennas'' as a generic term to designate telescope dishes, single antennas or beamformed subarrays.} receiving complex Gaussian cosmic signals $s(\bs l,t) \simiid \cl C\cl N(0,\skyvar(\bs l))$ from the sky with power flux density distribution $\intmap(\bs l) > 0$ [W/m$^2$], \ie the sky intensity distribution. We write $\bs l=(l,m)$ the direction cosines of the portion of the sky looked from the array formed by the antennas. More precisely, a \emph{direction cosine} coordinate system fixed with its origin at the center of the Earth is chosen such that the center of the field to be imaged is denoted by the unit vector $\bs s_0 = (l,m,n) = (0,0,1)$. The other directions in the region of interest are denoted by $\bs s=\bs s_0+\bs\tau$ with $\bs\tau = (l,m,\sqrt{1-l^2-m^2})$. The region of interest will be considered sufficiently small to approximate it as a plane\footnote{The invalidity of this assumption can be addressed with a spread spectrum model, \ie by considering the partial Fourier transform of the image with a linear chirp modulation~\cite{cornwell08,wiaux09,dabbech21}.}, or equivalently $\sqrt{1-l^2-m^2}\approx 1$. The set $\Omega(t):=\{\bs p_q^\perp(t)\}_{q=1}^Q \subset \Rbb^2$ denotes the projection onto the plane perpendicular to $\bs s_0$ of the instantaneous position of the $Q$ antennas---moving in time due to the rotation of the Earth. Our reasoning is inspired by the framework of~\cite{VanderVeen18} with the following adaptations:
\begin{enumerate}[label=G\arabic*.]
    \item \label{enum:21_G1} We assume a monochromatic signal with a single wavelength $\lambda$ and associated frequency $f=\frac{c}{\lambda}$ with the speed of light $c$. The separation into frequency subbands can be done efficiently with filter banks~\cite{VanderVeen18}.
    \item \label{enum:21_G2} We consider the signals at instantaneous time $t$ and thus report their sampling to later.
    \item \label{enum:21_G3} We deal with a continuous sky intensity distribution $\intmap(\bs l)$ rather than a finite number of signal sources. On the time scale of the acquisition, the intensity distribution is stationary.
    \item \label{enum:21_G4} We give the explicit expression of the phase factors $a_q(\bs l,t)=e^{\frac{\im2\pi}{\lambda} \bs p_q^\perp(t)^\top \bs l}$ that inform on the position-dependent geometric delays.
    \item \label{enum:21_G5} We consider the same direction-dependent gain $g(\bs l)$ for all antennas. 
\end{enumerate}

\begin{figure}[t]
\centering 
\includegraphics[width=.9\linewidth]{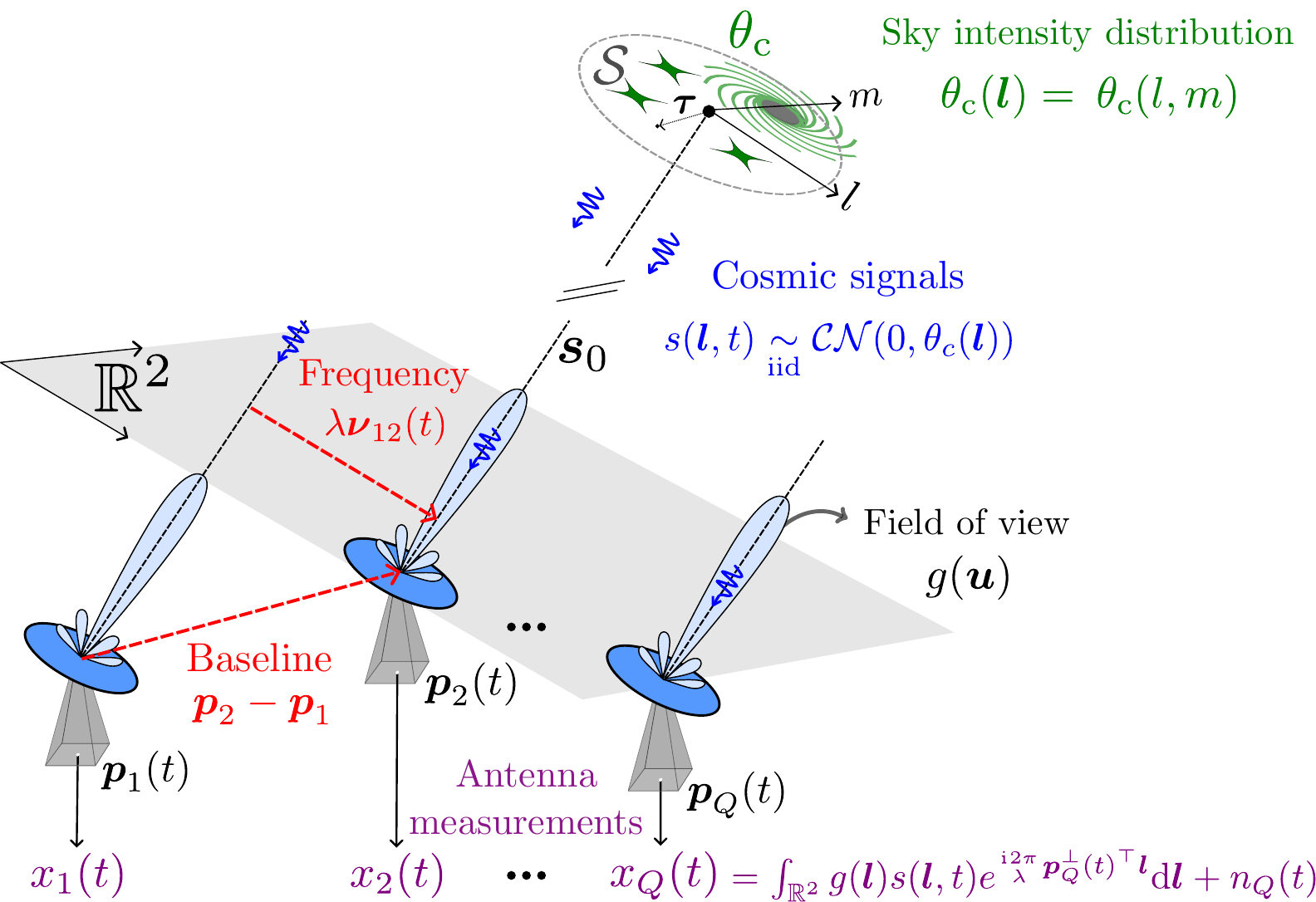}
\caption{Schematic of the radio-interferometry sensing context. Far-away cosmic signals following a Gaussian random process, \ie $s(\bs l,t)\sim_{\iid} \cl C\cl N(0,\skyvar(\bs l))$ with an intensity distribution $\intmap(\bs l)$, are received by $Q$ antennas. The antennas, positionned in $\Omega(t):=\{ \bs p_q^\perp(t) \}_{q=1}^Q$, have the same direction-dependent gain $g(\bs l)$ focusing a specific region $\cl S$ of the sky. Each antenna $q\in \upto{Q}$ integrates---with its own geometric delay $\bs p_q^\perp(t)^\top\bs l$---the cosmic signals from all direction into a noisy measurement $x_q(t)=\int_{\Rbb^2} g(\bs l)s(\bs l,t)e^{\frac{\im2\pi}{\lambda}\bs p_q^\perp(t)^\top\bs l} \ud\bs l + n_q(t)$.} 
\label{fig:21_schematic}
\end{figure}

Up to a compensation for the arrival time difference between individual antennas, we can always assume that all antennas lie in a plane perpendicular to the pointing direction $\bs s_0$.
Under the stated conditions,~\cite[Eq.~(9)]{VanderVeen18} which describes the temporal signal received by antenna $q \in \upto Q$ can be modified as $x_q(t) = \bar x_q(t) + n_q(t)$ with the noiseless measurements\footnote{The integration over $\Rbb^2$ is possible, despite $\bs l \in [0,1]^2$, because the gains $g$ will be zero outside $[0,1]^2$. }
\begin{equation} \label{eq:31_meas}
    \bar x_q(t) = \int_{\Rbb^2} g(\bs l) s(\bs l,t) e^{\frac{\im 2\pi}{\lambda} \bs p_q^\perp(t)^\top \bs l} \ud\bs l,
\end{equation}
where $s(\bs l,t) \simiid \cl{CN}(0,\skyvar(\bs l))$ and $\bs n(t) := (n_q(t))_{q=1}^Q \sim \cl{CN}(\bs 0,\bs\Sigma_{\bs n})$ are complex zero mean white Gaussian random processes with the assumption $\bb E [s(\bs l,t) s^*(\bs l',t)] = \skyvar(\bs l) \delta(\bs l-\bs l')$~\cite{VanderVeen18}.

We can stack the $Q$ received signals into an instantaneous measurement vector $\bs x(t) = (x_1(t),\ldots,x_Q(t))$, with the (Hermitian) \emph{covariance matrix} as $\covmat(t) := \bb E_s\E[\bs n]{\bs x(t) \bs x^*(t)} \in \cl H^Q$.
Leveraging the \emph{Van Cittert-Zernike} theorem~\cite{vancittert34,zernike38} and assuming that the same realization of the cosmic signals $s(\bs l,t)$ is received simultaneously at time $t$ for all antennas, the covariance matrix can be recast as 
\begin{equation} \label{eq:31_cov}
    \covmat(t) = \intMOmt{\vintmap} + \bs\Sigma_{\bs n}.
\end{equation}
In equation~\eqref{eq:31_cov}, $\bs\Sigma_{\bs n} := \E[\bs n]{\bs n(t)\bs n^*(t)}$ is the noise covariance and
\begin{equation} \label{eq:31_vign}
    (\intMOmt{\vintmap})_{jk} := \int_{\Rbb^2} \vintmap(\bs l) e^{\frac{-\im 2\pi}{\lambda} (\bs p_k^\perp(t)-\bs p_j^\perp(t))^\top \bs l} \ud\bs l
\end{equation}
is the $jk$-th entry of the \emph{interferometric matrix} of the map $\vintmap$---analogous to~\cite{leblanc24}---with the \emph{vignetted} map
$$
\vintmap(\bs l):=g^2(\bs l)\intmap(\bs l).
$$  
Overall,~\eqref{eq:31_cov} and~\eqref{eq:31_vign} show that RI corresponds to an interferometric system that is \emph{affine} in $\vintmap$. It is tantamount to sampling the 2-D Fourier transform of $\vintmap$ over frequencies selected in the difference multiset, 
\begin{equation} \label{eq:31_baselines}
\cl V(t) := \tinv{\lambda}(\Omega(t) - \Omega(t)) = \{ \bs \nu_{jk}(t) := \tfrac{\bs p_j^\perp(t) - \bs p_k^\perp(t)}{\lambda} \}_{j,k=1}^Q,
\end{equation}
\ie $(\intMOm{\vintmap}(t))_{jk} = \cl F[\vintmap](\bs \nu_{kj}(t))$, then adding the covariance matrix of the noise $\bs\Sigma_{\bs n}$.

\begin{figure}[t]
\centering 
\includegraphics[width=\linewidth]{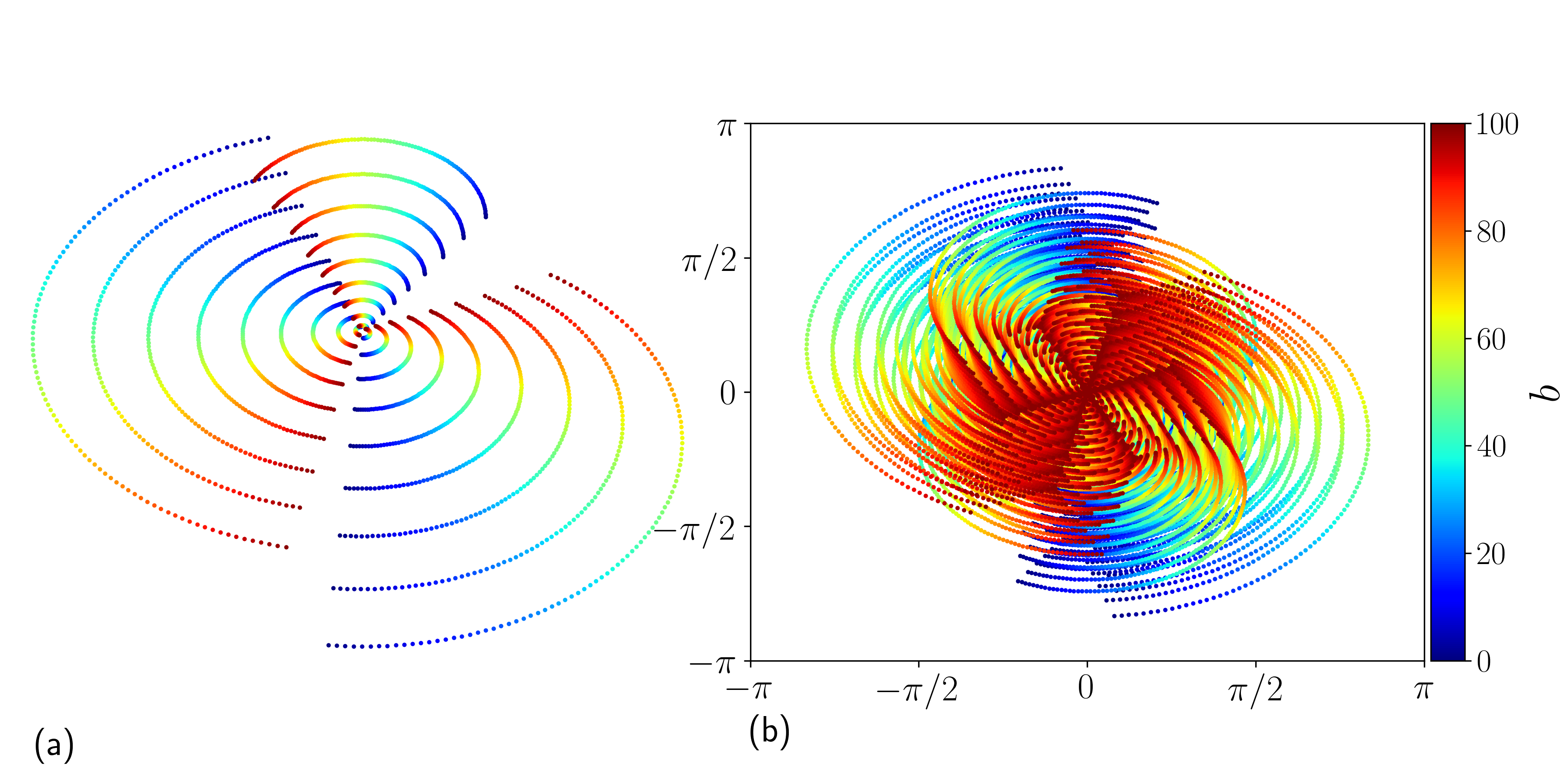}
\caption{(a) Arrangement $\{ \Omega_b \}_{b=1}^{100}$ of the antennas of the \emph{Very Large Array} (VLA)~\cite{thompson80vla} for a total observation time of five hours. (b) the visibility set $\{ \cl V_b \}_{b=1}^{100}$ corresponding to (a). The colors vary from blue for $b=1$ to red for $b=100$.} 
\label{fig:31_antennas_and_uv}
\end{figure}

In practice, the measurement vector $\bs x(t)$ must be time sampled as $\bs x[i]=\bs x(iT)$ with \emph{sampling period} $T$ and $i$ an integer to perform digital computations. Due to the rotation of the Earth, the time samples $\bs x[i]$ are separated into $B$ \emph{Short-Time Integration} (STI) intervals, or \emph{batches}, indexed by $b$ and of duration $IT$ with $I$ samples per STI. All these samples can be stacked into a \emph{signal set} $\cl X$ as
\begin{equation} \label{eq:xb}
    \ts\cl X := \bigcup_{b \in \upto{B}} \cl X_b,\quad \cl X_b :=\{\bs x_b[i],~i\in\upto{I}\},  
\end{equation}
where $\bs x_b[i] := \bs x[(b-1)I + i]$. 
With a typical sampling rate of $1$GHz, and an STI interval $IT$ of $15$s, $I \approx 15~10^9$ samples are enough to accurately approximate the covariance matrix with a \emph{sample covariance}. The batch duration is sufficiently short to assume that the set of antenna positions for batch $b$ remains nearly stationary within the time interval $\cl T_b := [bIT,(b+1)IT]$, \ie $\Omega_b := \{\bs p_q^\perp(t)\}_{q=1, t \in \cl T_b}^Q \approx \{ \bs p_{bq}^\perp \}_{q=1}^Q$. An example of realistic antenna arrangements $\{ \Omega_b \}\btoB$ and corresponding visibility sets $\{ \cl V_b \}\btoB$ for $\cl V_b := \Omega_b - \Omega_b$ is given in Fig.~\ref{fig:31_antennas_and_uv}. From the definition of $\bs x_b[i]$ in~\eqref{eq:xb}, the \emph{sample} covariance matrix of batch $b$ is defined as 
\begin{align} \label{eq:31_empi_covmat}
\begin{split}
    \sampcovmat_b(\cl X_b) &:= \tinv{I} \sum_{i=1}^I \bs x_b[i]\bs x_b^*[i] \\
    &\approx \E[s,\bs n]{\bs x_b[i] \bs x_b[i]^*} = \intMOmb[]{\vintmap} + \bs\Sigma_{\bs n},
\end{split}
\end{align}
for any $i\in\upto{I}$ and $\intMOmb[]{\vintmap} := \bc I_{\Omega(t=(b-1/2)MT)}[\vintmap]$. 

Anticipating over the \emph{boundedness} and \emph{bandlimitedness} assumptions \ref{h:bounded-FOV} and \ref{h:band-limitedness-fvign} given in Sec.~\ref{sec:recovery}, $\vintmap$ can be represented by a vectorized image $\intdisc \in \Rbb^N$ of $N$ pixels in the FOV. While the interferometric matrix at batch $b$, $\intMOmb[]{}$, can be modeled in matrix form as shown in Appendix~\ref{app:matrix}, the discrete representation of $\vintmap$ yields a vector formulation
\begin{equation} \label{eq:32_rb}
    \bar{\bs v}_b := \vect(\intMOmb[]{\intdisc}) = \bs G_b \bs F\intdisc
\end{equation}
where $\bs G_b \in \Cbb^{Q^2\times N}$ is a convolutional interpolation operator that interpolates the \emph{on-grid} frequencies obtained from the FFT $\bs F \intdisc $ in the continuous Fourier plane at frequencies corresponding to the difference set $\cl V_b := \cl V(t=(b-1/2)IT)$ defined in~\eqref{eq:31_baselines}. This procedure is known as \emph{Non-Uniform Fast Fourier Transform} (NUFFT). We practically use the \textsc{min-max} interpolation technique~\cite{fessler03}. Finally, as depicted in Fig.~\ref{fig:computations}(a), all the visibilities accumulated with the $B$ batches are combined as 
\begin{equation} \label{eq:32_r}
    \bar{\bs v} = \begin{bmatrix}
        \bar{\bs v}_1 \\ \vdots \\ \bar{\bs v}_B
    \end{bmatrix} = \begin{bmatrix}
        \bs G_1 \\ \vdots \\ \bs G_B
    \end{bmatrix} \bs F \intdisc  = \bs G \bs F \intdisc .
\end{equation}
Eq.~\eqref{eq:32_r} compares to the standard RI forward model~\cite[Eq.~(7)]{Pratley18} with less details about the NUFFT compensation terms such as zero-padding and gridding correction. The \emph{visibility operator} $\bs{GF}$ corresponds to the forward operator $\bs\Phi$ in the introduction. 

We emphasize the difference between the \emph{acquisition operator}
\begin{align*}
\label{eq:acqop}
    \cl X \mapsto\ & \acqop(\cl X):= (\acqop_1(\cl X_1)^\top,\ldots,\acqop_B(\cl X_B)^\top)^\top \in \Cbb^{Q^2B},\\
    &\text{with}\ \acqop_b(\cl X_b) := \vect(\bs C_b(\cl X_b)-\bs\Sigma_{\bs n})\in \Cbb^{Q^2},\nonumber
\end{align*} 
that computes the visibility vector $\bs v = \acqop(\cl X)$ from the antenna signals $\cl X$, and the \emph{imaging operator} 
\begin{equation} \label{eq:imop}
    \intdisc' \in \Rbb^N \mapsto \imop(\intdisc') := \bs G \bs F\intdisc'
\end{equation}
that computes a candidate visibility vector $\bar{\bs v}' = \imop(\intdisc')$ from a candidate image $\intdisc'$. These operators are identical up to a statistical noise, the origin of the approximation made in~\eqref{eq:31_empi_covmat}, which means that 
\begin{equation} \label{eq:oplink}
    \bb E_{\cl X} \acqop(\cl X) = \imop(\intdisc).
\end{equation}

Historically, the goal of radio-interferometric imaging has been to compute a good estimate $\wt{\intdisc}$ fitting the visibility data $\bs v$ yielded from $\acqop(\cl X(\intdisc))$. Among the latest state-of-the-art reconstruction algorithms, both SARA~\cite{carrillo12} and AIRI~\cite{terris22} algorithms aim to provide an estimate $\wt{\intdisc} \approx \intdisc$ by solving 
\begin{equation*}
    \wt{\intdisc} = \argmin_{\intdisc' \in \Rbb^N}~\tinv 2 \norm{\bs v-\bs G \bs F \intdisc'}{2}^2 + \lambda r(\intdisc')
\end{equation*}
where the term $r(\intdisc')$ is a specific \emph{regularization}. Unfortunately, the total number of visibilities\footnote{Half of the Hermitian matrix $\intMOmb[]{\intdisc}$ as well as its diagonal, containing the DC component, are usually removed from the measurements to avoid redundancy.} $V = Q(Q-1)B$ can become too large for arrays containing thousands of antennas and aggregating measurements over thousands of batches. The following sections describes ways to reduce this number of measurements.

\subsection{Post-Sensing Compression} \label{sec:aposteriori}

\begin{figure*}[t]
    \centering
    \includegraphics[width=0.9\linewidth]{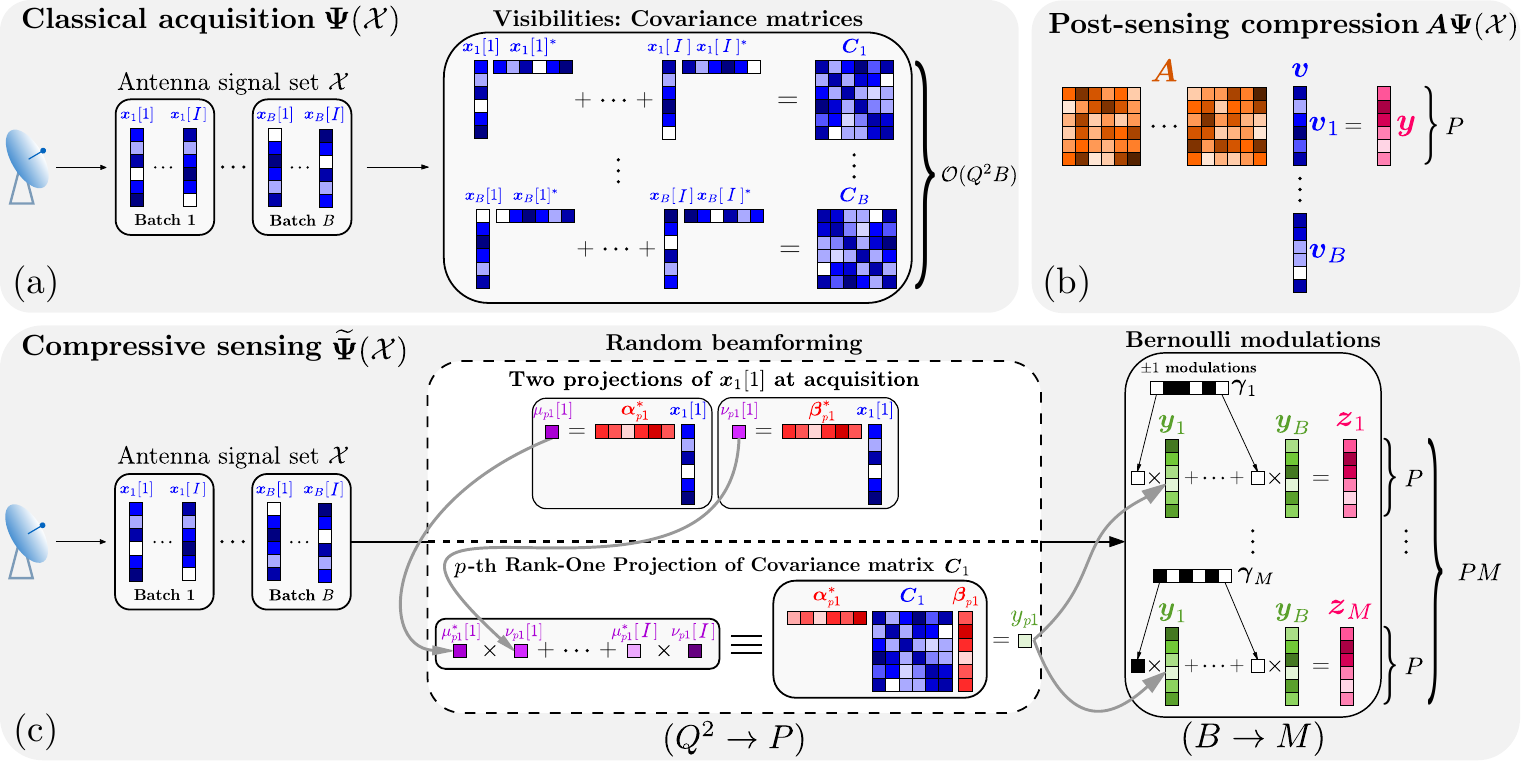}
    \caption{
    Computations at the acquisition: from the antenna signals $\cl X$ to the (compressed) visibilities. We consider the noiseless case ($\bs\Sigma_{\bs n}=\bs 0$) to simplify the illustration. (a) \textbf{Classical scheme}. For each batch $b$, the sample covariance matrix is computed as $\sampcovmat_b = \tinv I \sum_{i=1}^I \bs x_b[i] \bs x_b[i]^*$. The covariance matrices include all visibilities. (b) \textbf{Post-sensing compression}. The (vectorized) visibilities are compressed using an \iid random Gaussian matrix as $\bs y =\bs{Av}$. (c) \textbf{Compressive sensing}. For each batch $b$, $P$ ROPs of the covariance matrix are obtained as $y_{pb} = \tinv I \sum_{i=1}^I \bs\alpha_{pb}^* \bs x_b[i] \bs x_b[i]^* \bs\beta_{pb}$, $\forall p\in\upto{P}$. The $\{ \bs y_b \}\btoB$ ROP vectors are next \emph{modulated} with Bernoulli modulation vectors $\{ \bs\gamma_m \}_{m=1}^{M}$ and $\gamma_{mb} \simiid \cl U\{ \pm 1 \}$ as $\bs z_m := \sum\btoB \bs\gamma_m \bs y_b$. The $\pm 1$ values in the modulation vectors are represented by black and white cells.}
    \label{fig:computations}
\end{figure*}

\begin{figure}[t]
    \centering
    \includegraphics[width=0.95\linewidth]{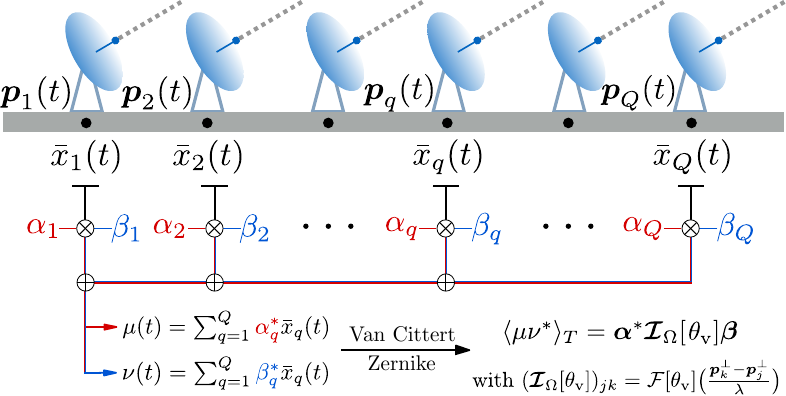}
    \caption{Projections of the (noiseless) measurement vector $\bar{\bs x}(t)$ with the sketching vectors $\bs\alpha$ and $\bs\beta$. By leveraging the Van Cittert-Zernike theorem, integrating the product of the projections $\mu(t)$ and $\nu(t)$ over time yields a ROP of the interferometric matrix of the vignetted image $\intMOm[]{\vintmap}$.}
    \label{fig:rafa}
\end{figure}

Several post-sensing compression techniques have been considered in the past. For instance, as depicted in Fig.~\ref{fig:computations}(b), one can reduce the dimension (and thus the memory footprint) of both operators $\acqop$ and $\imop$ with $P$ random Gaussian projections and compute, at the acquisition and at the reconstruction,
\begin{align} 
    \bs y &= \bs A \acqop(\cl X),~~\text{and} \label{eq:32_gaussian_acq} \\
    \bar{\bs y} &= \bs A \imop(\intdisc) = \bs{AGF} \intdisc, \label{eq:32_gaussian_im}
\end{align}
respectively, with $\bs y, \bar{\bs y} \in \Cbb^{P}$ and $\bs A \in \Cbb^{P \times Q^2B}$ with $A_{jk} \underset{\iid}{\sim} \cl{CN}(0,1/P)$~\cite{wiaux09}. 
It appears clearly in~\eqref{eq:32_gaussian_acq} that the compression is applied after the acquisition of the visibilities, justifying the ``post-sensing'' terminology.
Unfortunately, $\bs A$ is dense and makes its application untractable with $\cl O(P Q^2B)$ operations. 

Another possibility named \emph{Baseline-dependent averaging}~\cite{wijnholds18} consists in averaging the visibilities associated to low-frequency content over consecutive batches. This yields a reduced number of measurements, scaling as $\cl O(Q^2B')$, with an equivalent number of batches $B' \ll B$. Applying the \emph{averaging operator} $\bs S \in \{ 0,1 \}^{Q^2B'\times Q^2B}$ provides
\begin{align}
    \bs y &= \bs S \acqop(\cl X),~~\text{and} \label{eq:baseline_dep_av} \\
    \bar{\bs y} &= \bs S \imop(\intdisc) = \bs G' \bs F \intdisc  \label{eq:baseline_dep_av_im}
\end{align}
where $\bs G' \in \Cbb^{Q^2B'\times N}$ is the \emph{averaged} version of $\bs G$. Baseline-dependent averaging is cheap and can provide more than $80\%$ of compression~\cite{wijnholds18}. However, the resulting projection does not leverage the low-complexity representation of the image $\intdisc$.

\subsection{Compressive Sensing Scheme:\newline\phantom{C.\ } Random Beamforming and Random Modulations}

Here, we develop a two-layer compressive sensing model relying on \emph{(i)} random beamforming for compressing the information stored in the interferometric matrix of each batch, and \emph{(ii)} random modulations followed by time integration to further compress the data stream along the time domain. The so-called MROPs are presented as a good candidate for reduced memory storage and cheap acquisition. The following mathematical derivations focus on the noiseless measurements, and the impact of noise is discussed at the end of this section.

\paragraph{First layer: random beamforming}
In RI, \emph{beamforming} is a signal processing technique that has been used to enhance the sensitivity and resolution of radio telescopes by combining signals from multiple antennas~\cite{VanderVeen18}. Mathematically, given a sketching vector $\bs\alpha$ and the signals defined in~\eqref{eq:31_meas}, beamforming can be modeled as a projection (or \emph{sketch}) of the measurements vector as 
\begin{align} \label{eq:beamforming}
\begin{split}
    \textstyle \mu(t) := \bs\alpha^*\bs x(t) &\textstyle = \int_{\Rbb^2} g(\bs l) s(\bs l,t) \sum_{q=1}^Q \alpha_q^* e^{\frac{\im2\pi}{\lambda} \bs p_q^\perp(t)^\top \bs l}\, \ud\bs l \\
    &\textstyle = \int_{\Rbb^2} g_{\bs \alpha}(\bs l) s(\bs l,t)\, \ud\bs l
\end{split}
\end{align}
for an equivalent direction-dependent gain $g_{\bs \alpha}(\bs l):=g(\bs l) \sum_{q=1}^Q \alpha_q^* e^{\frac{\im2\pi}{\lambda} \bs p_q^\perp(t)^\top \bs l}$.
By adjusting the phase and amplitude of the signals from the $Q$ antennas using the sketching vector $\bs\alpha$, beamforming allows the array to focus on a specific direction in the sky by narrowing $g_{\bs \alpha}(\bs l)$ compared to $g(\bs l)$. 

As depicted in Fig.~\ref{fig:rafa}, let us consider two sketches $\mu(t)$ and $\nu(t)$ of the measurement vector computed from~\eqref{eq:beamforming} with the sketching vectors $\bs\alpha:=(\alpha_1,\ldots,\alpha_Q)$ and $\bs\beta:=(\beta_1,\ldots,\beta_Q)$, and more specifically their sampling $\mu_b[i]$ and $\nu_b[i]$ obtained in the same way as how $x_b[i]$ was obtained in~\eqref{eq:xb}. Aggregating their product in time gives
\begin{align} \label{eq:31_sample_rop}
\begin{split}
    \ts y_b &\ts = \tinv I \sum_{i=1}^{I} \mu_b[i] \nu_b^*[i] \\
    \ts &\ts = \tinv I \sum_{i=1}^{I} \bs\alpha^* \bs x_b[i] \bs x_b^*[i] \bs\beta = \bs\alpha^* \sampcovmat_b \bs\beta.
\end{split}
\end{align}
In~\eqref{eq:31_sample_rop}, $\bs\alpha^* \sampcovmat_b \bs\beta := \langle \bs\alpha\bs\beta^*, \sampcovmat_b \rangle_{\rm F}$ is called a \emph{rank-one projection} (ROP) of the sample covariance matrix because it amounts to projecting $\bs C_b$ with the rank-one matrix $\bs\alpha\bs\beta^*$~\cite{cai15,chen2015exact}. We refer to Sec.~\ref{sec:related} for appropriate references.
Eq.~\eqref{eq:31_sample_rop} showcases that the acquisition process has changed compared to Sec.~\ref{sec:model}.A-B. These random beamforming computations are illustrated in the center box in Fig.~\ref{fig:computations}(c). We will define a new acquisition operator after introducing \emph{layer 2}.

Inserting the definition of the sample covariance matrix~\eqref{eq:31_empi_covmat} into~\eqref{eq:31_sample_rop} yields 
\begin{equation} \label{eq:31_three_rop}
    y_b\ =\ \bs\alpha^* \sampcovmat_b \bs\beta\ \approx\ \bar{y}_b = \bs\alpha^* \intMOmb[]{\vintmap} \bs\beta + \bs\alpha^* \bs\Sigma_{\bs n} \bs\beta 
\end{equation}
where $\bar y_b$ is composed of two terms. First, the signal of interest, the ROP $\bs\alpha^* \intMOmb[]{} \bs\beta$ of the interferometric matrix (as detailed in App.~\ref{app:MCFLI}, this is analogous to~\cite[Eq.~(3)]{leblanc24} in MCFLI). Second, the quantity $\bs\alpha^* \bs\Sigma_{\bs n} \bs\beta$ is a \emph{fixed} bias term due to noise in the antenna measurements. This bias is expected to be compensated, at least partially, by the knowledge of the sketching vectors $\{\bs\alpha, \bs\beta\}$ and a good estimate of the noise covariance matrix $\bs\Sigma_{\bs n}$. We detail the analogy with MCFLI in App.~\ref{app:MCFLI}. 

We are going to show that one can devise a specific sensing scheme of the antenna signals that avoids storing individually the $B$ sample covariance matrices $\{ \sampcovmat_b \}\btoB$ as done classically in~\eqref{eq:31_three_rop}. This is possible while still ensuring an accurate estimation of the image of interest. This sensing first records $P \ll Q^2$ different random ROPs of each $\sampcovmat_b$, or, equivalently, $P$ evaluations of $y_b$ in~\eqref{eq:31_sample_rop} from different $\mu_b[i]$ and $\nu_b[i]$, associated with random vectors $\bs \alpha$ and $\bs\beta$. Following~\cite{leblanc24}, we consider $P$ \iid random \emph{sketching} vectors $\bs \alpha_{pb}, \bs \beta_{pb} \sim_{\iid} \bs \rho$, $p \in \upto{P}$, for some random vector $\bs\rho \in \bb C^Q$ with $\rho_q \simiid \exp(\im \cl U[0,2\pi)),~q\in\upto{Q}$. While random Gaussian vectors $\bs\alpha$ and $\bs\beta$ were another appropriate choice (the only condition is sub-Gaussianity), unitary vectors are more practical for implementation on real antennas, especially in analog operations, as they only require tuning the phase by beamforming techniques.

We focus here on the forward imaging model. The $P$ ROPs are gathered in the measurement vector $\bar{\bs y}_b := (\bar y_{pb})_{p=1}^{P} \in \bb C^{P}$, with $\bar y_{pb} = \bs\alpha_{pb}^* \intMOmb[]{\intdisc} \bs\beta_{pb}$. Moreover, following the methodology of Sec.~\ref{sec:meas}, we can write 
\begin{equation} \label{eq:32_yb}
    \bar{\bs y}_b = \bs R_b \bs G_b \bs F \intdisc ,
\end{equation}
where each row of the matrix $\bs R_b$ computes a single ROP measurement, \ie $(\bs R_b^\top)_p = \vect(\bs\alpha_{pb} \bs\beta_{pb}^*)^\top$, for all $p \in \upto{P}$. The $B$ ROPed batches can then be concatenated as 
\begin{equation} \label{eq:crop}
    \bar{\bs y} = \begin{bmatrix}
        \bar{\bs y}_1 \\ \vdots \\ \bar{\bs y}_B
    \end{bmatrix} = \begin{bmatrix}
        \bs R_1 & & \\ 
        & \ddots & \\ 
        & & \bs R_B
    \end{bmatrix} \begin{bmatrix}
        \bs G_1 \\ \vdots \\ \bs G_B
    \end{bmatrix} \bs F \intdisc  = \bs D\bs G \bs F \intdisc ,
\end{equation}
with $\bar{\bs y} \in \Cbb^{PB}$. The imaging model~\eqref{eq:crop}, simply consists of adding the \emph{concatenated ROP} operator $\bs D$ to the conventional model~\eqref{eq:32_r} sensing the visibilities. 

\paragraph{Layer 2: Random modulations}
Unfortunately, in~\eqref{eq:crop}, the number of measurements $PB$ still depends on the (fixed) number of batches $B$. It is also unclear if the composed operator $\bs{DGF}$ allows us to easily benefit from the union of visibilities at all batches, \ie a denser sampling of the image spectral domain.

Following an approach whose feasibility will be discussed momentarily, we thus apply $M$ random modulations of the ROP vectors before integrating them, \ie we compute 
\begin{equation} \label{eq:32_modul1}
    \ts \bar{\bs z}_m = \sum\btoB \gamma_{mb} \bar{\bs y}_b
\end{equation} 
with the \emph{modulation vectors} $\bs\gamma_m \simiid \bs\gamma$, $p \in \upto{M}$, and $\gamma_b \simiid \cl U\{\pm 1\}$, $b \in \upto{B}$. 

From the concatenated ROP model~\eqref{eq:crop}, this can also be viewed as 
\begin{equation*}
    \bar{\bs z}_m = \big[\gamma_{m1} \Id, \ldots, \gamma_{mB} \Id
    \big]\, \bar{\bs y} = (\bs\gamma_m^\top \otimes \Id) \bar{\bs y}
\end{equation*}
with $\Id \in \Rbb^{P\times P}$ and the Kronecker product $\otimes$. Moreover, the $M$ modulations can be concatenated as 
\begin{equation*} 
    \bar{\bs z} = \big[\bar{\bs z}_1^\top, \ldots, \bar{\bs z}_{M}^\top\big]^\top = (\bs\Gamma^\top \otimes \Id) \bar{\bs y} = \bs M \bar{\bs y},
\end{equation*}
with $\bs\Gamma := [\bs\gamma_1, \ldots, \bs\gamma_{M}] \in \{ \pm 1 \}^{B\times M}$, $\bs M := \bs\Gamma^\top \otimes \Id \in \{ \pm 1 \}^{P M \times P B}$, and $\bar{\bs z} \in \Cbb^{P M}$. The resulting imaging model writes 
\begin{equation} \label{eq:32_modul}
    \bar{\bs z} = \imopn(\intdisc) := \bs M\bs D\bs G \bs F \intdisc   
\end{equation}
which turns in to simply apply the \emph{modulation operator} $\bs M \in \{\pm1,0\}^{P M \times P B}$ to the concatenated ROP model in~\eqref{eq:crop}. 

On the acquisition side, this defines the \emph{compressive sensing operator} 
$$
 \ts\cl X \in \Cbb^Q \mapsto \bs z = ({\bs z_1}{}^\top, \ldots,  {\bs z_{M}}{}^\top)^\top := \acqopn(\cl X) \in \bb C^{P M},
$$
with 
$$
\ts z_{mp}  := \sum\btoB \gamma_{mb} \bs\alpha_{pb}^* (\sampcovmat_b(\cl X_b)-\bs\Sigma_{\bs n}) \bs\beta_{pb}^*,\ p \in \upto{P},\ m \in \upto{M}.
$$
As illustrated in Fig.~\ref{fig:computations}(c), this also means that, in the noiseless case ($\bs\Sigma_{\bs n}=\bs 0$) or up to a debiasing step removing the contribution of the noise covariance,  
$$
\ts z_{mp} := \sum\btoB \gamma_{mb} y_{pb},\ \text{with}\ y_{pb} = \tinv I \sum_{i=1}^I \mu_{pb}[i] \nu_{pb}^*[i]
$$
which involves the computation of antenna signal sketches $\mu_{pb}[i] := \bs\alpha_{pb}^* \bs x_b[i]$ and $\nu_{pb}[i] := \bs\beta_{pb}^* \bs x_b[i]$.

\begin{table}[!t] 
    \centering
    \caption{Computational cost of the acquisition and sample complexities.}
    \resizebox{\linewidth}{!}{
    \begin{tabular}{|c|c|c|c|}
        \hline
        Name & Model & Cost per batch & Max size \\
        \hline
        Classical acquisition & $\acqop$ & $\cl O(Q^2)$ & $\cl O(Q^2B)$ \\
        Post-sensing compression & $\bs A \acqop$ & $\cl O(P Q^2)$ & $\cl O(Q^2)$ \\
        Compressive sensing & $\acqopn$ & $\cl O(P Q)$ & $\cl O(P M)$ \\
        \hline
    \end{tabular}
    }
    \label{tab:complexity}
\end{table}

Table~\ref{tab:complexity} compares the computational cost and maximal data size of the classical, post-sensing compression and 
compressive sensing schemes. It is seen that the cost of the compressive sensing approach becomes smaller than the classical sensing scheme if $P < Q$. The key observation in Fig.~\ref{fig:computations}(c) is that the number of measurements never exceeds $P M$ during the acquisition because the ROP and modulations can be computed and aggregated on the fly. 

As shown in the numerical experiments in Sec.~\ref{sec:numerical}, $P M$ can be much smaller than the $V \simeq Q^2B$ visibilities of the classical scheme while still ensuring accurate image estimations. The random Gaussian post-sensing compression must compute an intermediate number $Q^2$ of visibilities during each batch. More particularly, computing $P$ projections of each visibility matrix costs $\cl O(P Q^2)$ operations per batch. 

In the compressive sensing scheme, we do not account for the number $\cl O(P B)$ of sketching elements needed for the random beamforming because these values can be selected and stored with less precision than the final data. \\

For the computational complexity of the forward imaging model, critical in all image reconstruction algorithms, the discussion about the structural differences between the operator $\bs{MD}$ composing the ROP and random modulations, and the operators $\bs A$ for the Gaussian compression and $\bs G'$ for the baseline-dependent averaging, is deferred to a parallel submission.

Finally, the MROPs find a natural interpretation of ROP applied to a block-diagonal interferometric matrix $\bc I$ gathering the $B$ intermediate interferometric matrices $\{ \intMOmb[]{} \}\btoB$. This interpretation reported to Appendix~\ref{app:rop} is at the core of the guarantees provided in Sec.~\ref{sec:recovery}.

\begin{remark}[Other noise sources] \label{rem:noise}
There are obviously many other noise sources than the bias related to the thermal noise at the antennas and written in~\eqref{eq:31_three_rop} such as quantum, quantization, and correlator noise~\cite[Chap.~1.3]{Thompson17} to name but a few. After compensation of the bias and assuming gaussian noise at the uncompressed visibilities~\eqref{eq:32_r}, it is possible to show that the noise model for the MROPs can be approximated as a centered Gaussian noise with a block-diagonal covariance matrix. Furthermore, if the covariance matrix of the noise at the visibilities is diagonal, so is the case for the noise impacting the MROPs.
\end{remark}

\section{Image Recovery Guarantees}
\label{sec:recovery}

Let us study the compressive imaging problem of directly estimating sparse images from the MROP imaging model~\eqref{eq:32_modul} given in Sec.~\ref{sec:model}. The acquisition process is not included in this section.

We simplify the analysis to the case where no modulation is applied before time aggregation, \ie $\gamma_{mb} = 1$ for all $b$ and a single $m$ in~\eqref{eq:32_modul1}, that is we follow the IROP model described by $\bar{\bs z} = \sum\btoB \bar{\bs y}_b \in \Cbb^P$. 
This simplification yields independent measurements, unlike MROPs, which is a condition necessary for our proofs.
In this case, the number of observations is only driven by $P$. 

From six simplifying assumptions made on both $\vintmap$ and the sensing scenario (see Sec.~\ref{sec:work-hyp}), we theoretically demonstrate that this method achieves reduced sample complexities compared to the uncompressed scheme of~\eqref{eq:32_r}. 

The IROP imaging model is equivalent to replacing $\bs\Gamma$ by $\bs 1_B$ so that $\bs M$ is replaced by $\bs 1_B^\top \otimes \Id = \big[\Id, \ldots, \Id\big]$ in~\eqref{eq:32_modul}, hence writing
\begin{equation}\label{eq:32_irop}
    \bar{\bs z} =  \big[\Id, \ldots, \Id\big]\,\bs D\bs G \bs F \intdisc  = \bs R \bs G \bs F \intdisc ,
\end{equation}
with $\bs R := \big[\bs R_1, \ldots, \bs R_B\big] \in \Cbb^{P \times Q^2B}$. As clear from~\eqref{eq:32_irop}, the ROP operator $\bs R$ is dense, and the goal of this section is to prove that the total forward operator $\bs{RGF}$ satisfies a variant of the usual \emph{restricted isometry property} (RIP), the $\RIPto$~\cite{chen2015exact}. This RIP, which relates the $\ell_2$-norm of the embedded vectors to the $\ell_1$-norm of their image by the forward operator, ensures the recovery of sparse images from their compressive observations in~\eqref{eq:32_irop}.
 
\subsection{Working Assumptions} \label{sec:work-hyp}

The theoretical guarantees for the IROP model rely on six realistic assumptions stated below.
We first assume a bounded field of view (FOV).
\begin{assumption}[Bounded FOV] \label{h:bounded-FOV} 
Given $L>0$, the support of the vignetting window $g^2(\bs l)$ is contained in a domain $\cl D := [-L/2, L/2] \times [-L/2, L/2]$, \ie $g=0$ on the frontier of $D$. 
\end{assumption}
\noindent Therefore, supposing $\intmap$ bounded, we have $\supp \vintmap \subset \cl D$ and $\vintmap = 0$ over the frontier of $D$. 
We also need to discretize $\vintmap$ by assuming it is bandlimited.
\begin{assumption}[Bounded and bandlimited image]
\label{h:band-limitedness-fvign} 
The sky image $\intmap$ is bounded, and $\vintmap$ is bandlimited with bandlimit $\frac{W}{2}$, with $W := \frac{N_1}{L}$ and $N_1 \in \bb N$, \ie $\cl F[\vintmap](\bs \chi) = 0$ for all $\bs \chi$ with $\|\bs \chi\|_\infty \geq \frac{W}{2}$.
\end{assumption}
As will be clear below, this assumption enables the computation of the total interferometric matrix $\intMOm[]{\vintmap}$ from the discrete Fourier transform of the following discretization of $\vintmap$. From \ref{h:bounded-FOV} and \ref{h:band-limitedness-fvign} the function $\vintmap$ can be identified with a vector $\intdisc \in \bb R^N$ of $N=N_1^2$ components. Up to a pixel rearrangement, each component $\intdiscel_j$ of $\intdisc$ is related to one specific pixel of $\vintmap$ taken in the $N$-point grid
\begin{equation*}
\ts \cl G_N := \frac{L}{N_1} \{(s_1,s_2)\}_{s_1,s_2=-\frac{N_1}{2}}^{\frac{N_1}{2}-1} \subset \cl D.
\end{equation*}
The Discrete Fourier Transform (DFT) $\hat{\intdisc}$ of $\intdisc$ is then computed from the 2-D DFT matrix $\bs F \in \bb C^{N \times N}$, \ie $\hat{\intdisc} = \bs F \intdisc \in \bb C^N,\quad\bs F := \bs F_1 \otimes \bs F_1$,
with $(\bs F_1)_{kl} =  e^{-\frac{\im2\pi}{N_1} kl}/\sqrt{N_1}$, $k,l \in \upto{N_1}$, and the Kronecker product $\otimes$. Each component of $\hat{\intdisc}$ is related to a 2-D frequency of
\begin{equation*}
\ts \hat{\cl G}_N := \frac{W}{N_1} \{\chi_1, \chi_2\}_{\chi_1,\chi_2=-\frac{N_1}{2}}^{\frac{N_1}{2}-1} \subset [-\frac{W}{2},\frac{W}{2}] \times [-\frac{W}{2},\frac{W}{2}].
\end{equation*}
 
Assumptions \ref{h:bounded-FOV} and \ref{h:band-limitedness-fvign} also allow us to compute the \emph{total interferometric matrix} $\intMOm[]{\vintmap} \in \cl H^{QB}$ defined in Appendix~\ref{app:rop} whose $b$-th diagonal block can be computed from $\intdisc$ as 
\begin{equation} \label{eq:equiv-cont-discrt-interfero}
  \intMOmb[]{\vintmap} = \bs W_b \bs F \bs D_{\intdisc} \bs F^* \bs W_b^*,
\end{equation}
as shown in App.~\ref{app:matrix}. This matrix $\intMOmb[]{\vintmap}$ will be useful to interpret the IROP measurements as ROP applied to it.

We need now to simplify our selection of the visibilities.  
\begin{assumption}[Distinct nonzero visibilities]
\label{h:distinct-visib}
Defining the total visibility set $\cl V_0 = \bigcup\btoB \cl V_b \setminus \{\bs 0\}$ for the visibilities $\cl V_b := \{ \frac{\bs p_{bj}^\perp - \bs p_{bk}^\perp}{\lambda} \}_{j,k=1}^Q$ defined in Sec.~\ref{sec:model}, we assume that all nonzero visibilities in $\cl V_0$ are unique, which means that $V = |\cl V_0| = \sum\btoB |\cl V_b| = Q(Q-1)B$. 
\end{assumption} 

If $\vintmap$ has zero mean, $\hat\intdiscel_0 = (\bs F \intdisc )_0 = 0$, the diagonal entries of $\intMOm[]{}$ are zero and
\begin{equation} \label{eq:equiv-frob-l2}
\ts \norm{\intMOm{\vintmap}}{\rm F}^2 = \norm{ \bs G_0 \bs F \intdisc }{}^2,
\end{equation}
where $\bs G_0$ interpolates on-grid frequencies to the continuous frequencies defined in the set $\cl V_0$, related to the off-diagonal entries of $\intMOm{\vintmap}$ and whose rows are all different (from~\ref{h:distinct-visib}).

We need to \emph{regularize} the (ill-posed) imaging problem by supposing that $\intdisc$ is \emph{sparse}. For the same reasons as given in~\cite[Sec.~VI]{leblanc24}, the sparsity of $\intdisc$ is restricted to the canonical basis. 
\begin{assumption}[Sparse image]
\label{h:sparsifying-basis}
The discrete image $\intdisc$ is $K$-sparse (in the canonical basis):
$\intdisc \in \Sigma_K := \{\bs v: \|\bs v\|_0 \leq K\}$.
\end{assumption}

Our theoretical analysis leverages the tools of compressive sensing theory~\cite{candes2006a,foucart17}. In particular, as stated in the next assumption, we require that the total interferometric matrix---actually, its non-diagonal entries encoded in $\cl V_0$---captures enough information about any sparse image $\intdisc$. 
\begin{assumption}[$\RIP$ for visibility sampling]
\label{h:rip-visibility}
Given a sparsity level $K$, a distortion $\delta > 0$, and provided that the total number of visibilities $V$ satisfies
\begin{equation} \label{eq:34_sample_comp}
V \geq \delta^{-2} K\, {\rm plog}(N,K, \delta),
\end{equation}
for some polynomials ${\rm plog}(N,K,1/\delta)$ of $\log N$, $\log K$ and $\log 1/\delta$, the matrix $\bs G_0 \bs F \intdisc $ defined\footnote{the index $0$ indicates that the DC component is excluded from the Fourier sampling. Only the frequencies at $\cl V_0$ are computed.} in~\eqref{eq:32_r} respects the $\RIP(\Sigma_K,\delta)$, \ie
\begin{equation*} \label{eq:34_rip-def}
\ts (1-\delta) \norm{\bs v}{}^2 \leq \frac{N}{\varpi^2 V} \norm{\bs G_0 \bs{Fv}}{2}^2 \leq (1+\delta) \norm{\bs v}{}^2,\ \forall \bs v \in \Sigma_K
\end{equation*}
\end{assumption}
\noindent where $\varpi := \frac{L^2}{\sqrt N}$ can be found from the continuous interpolation formula of the Shannon-Nyquist sampling theorem.
As will be clear later, combined with~\eqref{eq:equiv-frob-l2}, this assumption ensures that two different sparse images lead to two distinct interferometric matrices, a key element for stably estimating images from our sensing model (see Prop.~\ref{prop:L2L1}). 

We specify now the distribution of the sketching vectors $\bs\alpha$, $\bs\beta$.
\begin{assumption}[Random sub-Gaussian sketches] \label{h:sketch-distrib}
The sketching vectors $\{\bs\alpha_p,\bs\beta_p\}_{p=1}^{P}$ involved in~\eqref{eq:32_irop} have \iid sub-Gaussian components. 
\end{assumption}
Assumption~\ref{h:sketch-distrib} is required to use Lemma~\ref{lem:27_concentration_ROP_A} in order to prove Prop.~\ref{prop:24_rip_rop}. We refer to~\cite[Sec. IV.B]{leblanc24} for the rationale and limitations of our assumptions.

\subsection{Data Centering}
\label{sec:centering}

The IROP model provides noisy measurements $\bs z$ written as 
\begin{equation} \label{eq:34_irop}
    \ts \bs z = \bs R\bs G \bs F \intdisc  + \bs\xi
\end{equation} 
for an additive noise $\bs\xi$ and after a debiasing step removing the measurement noise bias $\bs R \vect(\bs \Sigma_{\bs n})$ (see~\eqref{eq:31_three_rop}). Because of the multiplicity $Q$ of the \emph{mean} (\aka DC component) in the sensing operator $\bs{GF}$, only its \emph{centered version} $\bs G_0 \bs F$---removing the $Q$ DC samples---can satisfy a $\RIP$, as stated in Assumption~\ref{h:rip-visibility}. The DC component $\hat\intdiscel_0=(\bs F \intdisc )_0$ can be easily estimated from the autocorrelation of the measurements at a single antenna. From this, the contribution of the DC component at each batch $b$ can be subtracted from the measurements as
\begin{equation} 
    \bs z_b^\rmc := \bs z_b - \hat\intdiscel_0~ \bs R_b \bs G_b \bs e_1.
\end{equation}
With this DC compensation, a \emph{centered} version of the IROP model~\eqref{eq:34_irop} is considered as 
\begin{equation} \label{eq:34_irop_c}
    \bs z^{\rm c} = \bs{RG}_0 \bs F \intdisc + \bs\xi^\rmc
\end{equation} 
The centered model~\eqref{eq:34_irop_c} thus senses, through $\bs G_0 \bs F \intdisc $, the off-diagonal elements of the total interferometric matrix $\intMOm[]{\vintmap}$. We will show below that the combination of $\bs R$ with the interferometric sensing respects a variant of the RIP property, thus enabling image reconstruction guarantees.

\subsection{Reconstruction analysis} \label{sec:rec-anl}

It is possible to compute an estimate $\wt{\intdisc}$ of $\intdisc$ from its measurements~\eqref{eq:34_irop_c} (whose sensing is valid under assumptions \ref{h:bounded-FOV}-\ref{h:band-limitedness-fvign}). This can be done by solving a variant of the basis pursuit denoise program~\cite{foucart17}, or BPDN$_{\ell_1}$, where the measurement fidelity constraint is defined with an $\ell_1$-norm, \ie
\begin{equation} \label{eq:BPDN}
  \ts \wt{\intdisc} = \argmin_{\bs v \in \Rbb^N} \norm{\bs v}{1} \st 
  \norm{ \bs z^\rmc - \CRIop \bs v}{1} \leq \epsilon, \tag{BPDN$_{\ell_1}$}
\end{equation}
where one must impose $\epsilon \geq \norm{\bs\xi^{\rm c}}{1}$ to reach feasibility.

This specific constraint in \ref{eq:BPDN} is adjusted to the properties of the ROP operator $\bs R$. We indeed show below that $\CRIop$ respects an adapted RIP, the $\RIPto(\Sigma_K,{\sf m}_K,{\sf M}_K)$: given a sparsity level $K$, and two constants $0<{\sf m}_K<{\sf M}_K$, this property imposes
\begin{equation} \label{eq:RIP-L2L1-def}
\ts {\sf m}_K \norm{\bs v}{2} \leq \tinv{P} \norm{\CRIop \bs v}{1} \leq {\sf M}_K \norm{\bs v}{2}, \quad \forall \bs v \in \Sigma_K.
\end{equation}
In a nutshell, if the image $\intdisc$ is $K$-sparse and the IROP forward operator $\CRIop$ satisfies the $\RIPto$ over $\Sigma_{2K}$, then, given any $K$-sparse estimate $\widetilde{\intdisc}$ of $\intdisc$, by setting $\bs v$ above to the $2K$-sparse residue $\intdisc - \widetilde{\intdisc}$, one observes that imposing a small $\ell_1$-distance between the image and its estimate \emph{in the measurement space} of $\CRIop$ amounts to restricting their actual $\ell_2$-distance.

The following proposition (inspired by~\cite[Lemma 2]{chen2015exact}) formalizes this connection and shows that the error $\norm{\intdisc -\wt{\intdisc}}{2}$ is bounded. 
\begin{proposition}[$\ell_2/\ell_1$ instance optimality of \ref{eq:BPDN}]
\label{prop:L2L1}
Given $K$, if there exists an integer $K' > 2K$ such that, for $k \in \{K', K + K'\}$, the operator $\CRIop$ satisfies the $\RIPto(\Sigma_k, {\sf m}_{k}, {\sf M}_{k})$ for constants $0 < {\sf m}_{k} < {\sf M}_{k} < \infty$, and
\begin{equation} \label{eq:RIP-L2L1-bound-condition-BPDN} 
    \tinv{\sqrt 2} {\sf m}_{K+K'} - {\sf M}_{K'} \sqrt{\tfrac{K}{K'}} \geq \gamma >0,
\end{equation}
then, for all $\intdisc$ sensed through $\bs z^\rmc = \CRIop \intdisc + \bs\xi^{\rm c}$ with bounded noise $\norm{\bs\xi^{\rm c}}{1} \leq \epsilon$, the estimate $\wt{\intdisc}$ provided by \ref{eq:BPDN} satisfies
\begin{equation} \label{eq:bpdn-inst-opt}
    \norm{ \intdisc - \wt{\intdisc}}{2} \leq C_0 \tfrac{\|\intdisc -\intdisc_K\|_1}{\sqrt{K'}} + D_0 \tfrac{\epsilon}{P},
\end{equation}
for two values $C_0=\cl O({\sf M}_{K'}/\gamma)$ and $D_0=\cl O(1/\gamma)$.
\end{proposition}
\begin{proof}
  The proof follows the proof given in~\cite[App.~B]{leblanc24}.
\end{proof}
The inequality~\eqref{eq:bpdn-inst-opt} is known as the \emph{instance optimality} of the \ref{eq:BPDN} estimate $\wt{\intdisc}$~\cite{foucart17}. This result expresses that the reconstruction error between the estimate $\wt{\intdisc}$ and the actual image $\intdisc$ is bounded by the sum of two terms: a first term that measures the deviation from the sparsity model---the sparser $\intdisc$ is, the smaller it deviates from its $K$-sparse approximation $\intdisc_K$ in the approximation error $\|\intdisc - \intdisc_K\|_1$---, and a second term that is proportional to the level of noise in the measurements. 

The technical condition~\eqref{eq:RIP-L2L1-bound-condition-BPDN}, which ensures the positivity of both $C_0$ and $D_0$ in~\eqref{eq:bpdn-inst-opt}, is satisfied if 
\begin{equation}
\ts  K'  > 8 \Big(\frac{{\sf M}_{K'}}{{\sf m}_{K+K'}}\Big)^2 K. \label{eq:bound-condition-bpdn-reloaded}
\end{equation}
In this case, $\gamma = \tinv{2\sqrt 2} {\sf m}_{K+K'}$, and, from~\cite[App.~B]{leblanc24}, $C_0 = 2(\sqrt 2+1) ({\sf M}_{K'}/{\sf m}_{K+K'})+2$ and $D_0 =  4(\sqrt 2+1)/{\sf m}_{K+K'}$. 

Interestingly, if both $P$ and $V$ sufficiently exceed the image sparsity level, the operator $\CRIop$ respects, with high probability, the $\RIPto$ required by Prop.~\ref{prop:L2L1}.
\begin{proposition}[$\RIPto$ for $\CRIop$ using asymmetric ROP] \label{prop:24_rip_rop}
Assume that assumptions \ref{h:bounded-FOV}-\ref{h:sketch-distrib} hold, with \ref{h:rip-visibility} set to sparsity level $\splev >0$ and distortion $\delta$ over the set $\Sigma_\splev$. For some values $C, c >0$, if   
\begin{equation} \label{eq:34_QM_vs_K2}
\ts P \geq C \splev \ln(\frac{12 e N}{\splev}),\ V \geq 4 \splev\, {\rm plog}(N,\splev),
\end{equation} 
then, with probability exceeding $1 - \exp(-c P)$, the operator $\CRIop$ respects the $\RIPto(\Sigma_\splev, {\sf m}_\splev, {\sf M}_\splev)$ with
\begin{equation} \label{eq:rip-l2l1}
\ts m_\splev > \frac{2}{3} c_1 \varpi \sqrt{\tfrac{V}{2N}},\ \text{and}\ M_\splev < \frac{4}{3} c_2 \varpi \sqrt{\tfrac{3V}{2N}},
\end{equation}
where the constants $c_1$ and $c_2$ depend only on the sub-Gaussian norm of the random sketching vectors $\{\bs\alpha_p \}_{p=1}^{P}$ and $\{\bs\beta_p \}_{p=1}^{P}$ hidden in $\bs R$ (see~\eqref{eq:32_irop}).
\end{proposition}
In this proposition, proven in Appendix \ref{app:proof}, while the bounds in~\eqref{eq:rip-l2l1} might not be tight, they cannot be arbitrarily close either. Still, the instance optimality~\eqref{eq:bpdn-inst-opt} holds as long as $c_2/c_1 = \cl O(1)$ since, thanks to~\eqref{eq:bound-condition-bpdn-reloaded},~\eqref{eq:RIP-L2L1-bound-condition-BPDN} can be verified when $K'>K$ is large enough. Indeed, combining Props.~\ref{prop:L2L1} and \ref{prop:24_rip_rop} and using the bounds in~\eqref{eq:rip-l2l1} that are independent of $\splev$, since $8 ({\sf M}_{K'}/{\sf m}_{K+K'})^2 < \frac{2\sqrt{3} c_2}{c_1}$,~\eqref{eq:bound-condition-bpdn-reloaded} holds if 
$$
K'> K \tfrac{2\sqrt{3} c_2}{c_1}.
$$
Therefore, provided the IROP sensing $\CRIop$ satisfies the $\RIPto(\Sigma_\splev,{\sf m}_\splev, {\sf M}_\splev)$ for the levels $\splev\in\{K',K+K'\}$, the instance optimality~\eqref{eq:bpdn-inst-opt} holds with 
\begin{equation*}
  \resizebox{\linewidth}{!}{$
\ts C_0 <  2(1+\sqrt 2) \sqrt{\frac{c_2}{c_1}} = \cl O(1),\ D_0 = \cl O\Big(\frac{\sqrt{N}}{\varpi \sqrt{V}}\Big) = \cl O\big(\frac{N}{L^2 Q \sqrt{B}}\big)$.
}
\end{equation*}
The condition on $K'$ and $K$ above has a direct impact on the number of IROPs $P$ and visibilities $V$, \ie on the \emph{sample complexities} of the sensing model~\cite{foucart17}. While $P$ and $V$ must be set large enough in~\eqref{eq:34_QM_vs_K2} when $K_0$ is set to $\splev=(K+K')> (\frac{2c_2 \sqrt{3}}{c_1} +1) K$, the impact of the sparsity error $\|\intdisc - \intdisc_K\|_1$ in~\eqref{eq:bpdn-inst-opt} is, however, attenuated by $1/\sqrt{K'} < \sqrt{c_1} / (\sqrt{2c_2 K})$. 

For a fixed FOV $L^2$, we also observe a meaningful amplification of the noise error by $D_0$ when the sampling grid $\cl G_N$ is too large compared to $Q \sqrt{B}$: if the number of pixels $N$ is too small,~\ref{h:band-limitedness-fvign} may not be verified, since the image bandwidth lower bounds $N$; if $N$ is too large the noise error in~(\ref{eq:bpdn-inst-opt}) is vacuous.

\section{Numerical Analysis} \label{sec:numerical}

\begin{figure*}[t]
    \centering   
    \includegraphics[width=\linewidth]{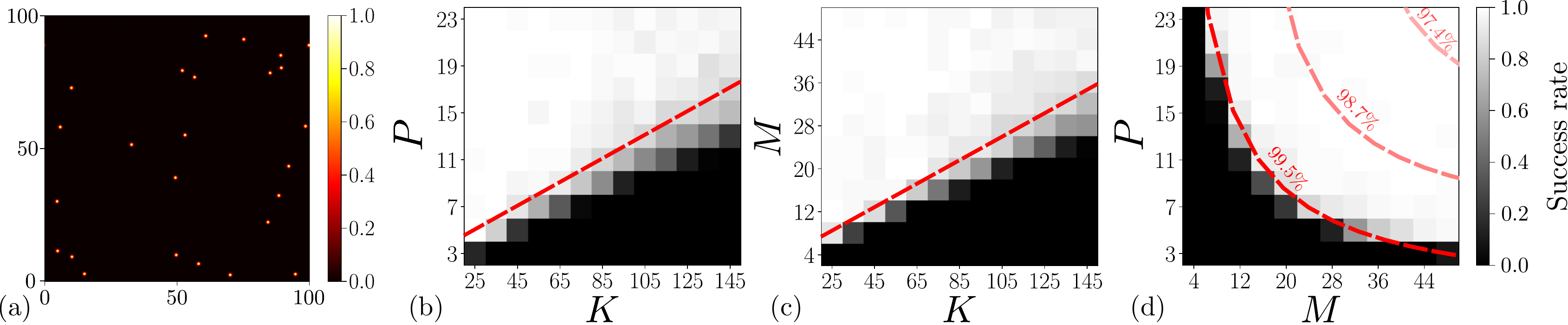}
    \caption{(a) Example of a randomly generated $K$-sparse sky with $K=25$. (b-d) Phase transition diagrams showing $P\times M$ MROPs of $B=100$ different $27 \times 27$ interferometric matrices for a $K$-sparse image $\intdisc$ (with $M=50$ in (b), $P=25$ in (c), and $K=25$ in (d)). One considers the $uv$-coverage shown in Fig.~\ref{fig:31_antennas_and_uv}, ROP using $ \bs\alpha_{pb}$, $\bs\beta_{pb}$ with $(\bs\alpha_{pb})_q,~(\bs\beta_{pb})_q \simiid e^{\im \cl U[0,2\pi)}$, and Bernoulli modulation vectors $\bs\gamma_{mb} \simiid \{\pm 1\}$, $\forall b\in\upto{B},~p\in\upto{P},~m\in\upto{M}$. Each pixel is constructed with $S=80$ reconstruction trials solving~\eqref{eq:35_bpdn} where we consider success if SNR$\ge 40$dB. The probability of success ranges from black (0\%) to white (100\%). Dashed red lines highlight the transition frontiers. The level curves with compression factors ($97.4\%$, $98.7\%$, $99.5\%$) are given in (d).}
    \label{fig:35_transmat}
\end{figure*}

We performed extensive Monte Carlo simulations with $S$ trials on the reconstruction of sparse images $\intdisc \in \Sigma_K$ with $N=100\times100$ pixels from the (noiseless) MROP imaging model $\bar{\bs z} = \bs M\bs D\bs G \bs F \intdisc $ derived in Sec.~\ref{sec:model}. The basis pursuit denoise program with an $\ell_2$-norm fidelity (or BPDN$_{\ell_2}$) is used to compute the estimate $\wt{\intdisc}$ as 
\begin{equation} \label{eq:35_bpdn}
    \wt{\intdisc} = \argmin_{\intdisc'}~ \norm{\intdisc'}{1} ~\st \norm{\bar{\bs z} - \bs M\bs D\bs G \bs F \intdisc' }{2} \leq 10^{-2}.
\end{equation}

\begin{remark}
    Using SPGL1\footnote{(Python module: \url{https://github.com/drrelyea/spgl1}).},~\eqref{eq:35_bpdn} is solved in its equivalent unconstrained formulation~\cite{pareto} using the \emph{proximal gradient} (\aka forward-backward) algorithm~\cite{beck_book,Parikh_Boyd}.
\end{remark}

While the constrain is imposed in the $\ell_2$-norm hence deviates from the theoretical setting established in Sec.~\ref{sec:recovery} that uses the $\ell_1$-norm, it avoids solving an internal minimization problem for computing the proximal operator of that constrain. Despite these differences, we provide similar conclusions than in Sec.~\ref{sec:recovery} for the sample complexity, as shown below.
In a noiseless setting (\ie $\bs\xi = \bs 0$), BPDN$_{\ell_2}$ is equivalent to BPDN$_{\ell_1}$.
Indeed, for a given $\bs z$ and $\bs B := \bs{MDGF}$, defining the feasibility constraint sets $\cl C := \{\intdisc \in \Rbb_+: \bs z = \bs B \intdisc \}$ , $\cl C_q(\epsilon) := \{\intdisc \in \Rbb_+: \norm{\bs z - \bs B \intdisc }{q} < \epsilon\}$ associated with the basis pursuit (BP) and BPDN$_{\ell_q}$ programs (for $q \in \{1,2\}$), the noiseless setting corresponds to the matching limit cases $\lim_{\epsilon \rightarrow 0}~ \cl C_1(\epsilon) = \lim_{\epsilon \rightarrow 0}~ \cl C_2(\epsilon) = C$.

We consider the reconstruction of a 2-D image supported in a square window so that one can assume $\vintmap=\intmap$, and satisfying assumptions \ref{h:bounded-FOV} and \ref{h:band-limitedness-fvign}. A vectorized image $\intdisc \in \Rbb^N$ was generated with a $K$-sparse support picked uniformly at random in $\upto{N}$, its $K$ nonzero components being all set to $1$. An example of a sparse sky\footnote{the image shown has been slightly blurred with a Gaussian kernel to enhance visual appeal.} image with $K=25$ is shown in Fig.~\ref{fig:35_transmat}(a). The partial Fourier sampling induced by the NUFFT operator $\bs{GF}$ is fixed by a realistic $uv$-coverage of the VLA~\cite{thompson80vla} shown in Fig.~\ref{fig:31_antennas_and_uv} with $Q=27$ antennas and $B=100$ batches corresponding to a total integration time of $5$ hours. 
At each simulation trial, we used $P B$ sketching vectors $\bs\alpha_{pb}, \bs\beta_{pb} \in \Cbb^Q$, $\forall p\in\upto{P}$, \iid sub-Gaussian\footnote{Similar results were obtained with random Gaussian sketching vectors.} with $(\bs\alpha_{pb})_q \simiid e^{\im \cl U[0,2\pi)}$, $q \in \upto{Q}$, and similarly for $\bs\beta_{pb}$. This choice is more practical for implementation on real antennas as it only requires tuning the phase by beamforming techniques. The random modulation vectors $\{ \bs\gamma_m \}_{m=1}^{M}$ were randomly picked as $\bs\gamma_m \simiid \cl U\{\pm1\}$. 

In Fig.~\ref{fig:35_transmat}, the success rates---\ie the percentage of trials where the reconstruction SNR exceeded $40$dB---were computed for $S=80$ trials per value of $(K,P,M)$, and for a range of $(K,P,M)$ specified in the axes. 

We observe in Fig.~\ref{fig:35_transmat}(b-d) that high reconstruction success is reached as soon as the number of MROPs satisfies $P M \geq C K$, with $C\simeq 5$. This is closely related to the sample complexity obtained for the IROP scheme in Prop.~\ref{prop:24_rip_rop}, where $P \geq C'K$ projections were needed (up to log factors). Here, the transition diagrams seem to indicate that $P$ and $M$ play the same role in the sample complexity, with only the product $P M$ mattering. This is also confirmed in Fig.~\ref{fig:35_transmat}(d) where the transition frontier in red describes a hyperbola at $(P, M)$ coordinates satisfying $P M \approx 150$. This implies that, for $K=25$, image recovery can be obtained with a number of MROPs which is only $0.2\%$ of the number of visibilities $V=70200$. A comparison between the sample complexity $P M$ needed for the MROPs and $P$ needed for the IROPs would have been compelling, but the IROP model is impossible to compute because of a much higher computational cost, discussed in another upcoming publication.

\section{Conclusion and Perspectives} \label{sec:conclusion}

In this paper, we focused on the acquisition process of radio-interferometry and proposed a \emph{compressive sensing} technique that can be applied at the level of the antennas. The novelty of the proposition was to show that \emph{random beamforming} is tantamount to applying ROPs of the covariance matrix containing the visibilities, and that these ROPs can be efficiently aggregated across time by random modulations. We provided recovery guarantees for the IROP model and observed the derived sample complexities under numerical conditions for the MROP imaging model. 

This paper presented a detailed derivation of the acquisition and imaging models, focusing specifically on reconstructing sparse images from the imaging models.
The sample complexity of the IROP model has additionally been characterized through the framework of compressive sensing theory.
However, several important aspects related to realistic imaging were left beyond the scope of this analysis. A subsequent paper will address key questions, such as the computational complexity of the imaging models with the integration of visibility weighting into the MROP framework, and a more comprehensive evaluation of reconstructions for realistic sky images. That work will center on practical considerations and applications in astrophysical imaging.

\section*{Acknowledgment}

Computational resources have been provided by the supercomputing facilities of UCLouvain (CISM) and the Consortium des Equipements de Calcul Intensif en Fédération Wallonie Bruxelles (CECI) funded by the Fond de la Recherche Scientifique de Belgique (FRS-FNRS). This project has received funding from FRS-FNRS (QuadSense, T.0160.24). The research of YW was supported by the UK Research and Innovation under the EPSRC grant EP/T028270/1 and the STFC grant ST/W000970/1.

\begin{appendices}

\section{Connections to MCFLI} \label{app:MCFLI}

We here report the commonalities and differences between compressive radio-interferometry (CRI) and the sensing model supporting multi-core fiber lensless imaging (MCFLI)~\cite{leblanc24}. Let us first recall the expression of a single-pixel measurement in MCFLI:
\begin{equation} \label{eq:31_mcfli_meas}
    y = \bs\alpha^* \intMOm{wf} \bs\alpha + n,
\end{equation}
where $f$ represents a continuous image (\eg from a biological sample), $w$ is the field-of-view window induced by each core aperture, and $n$ is a measurement noise.
Eq.~\eqref{eq:31_mcfli_meas} must be compared with the ROP model for radio-interferometric measurements in~\eqref{eq:31_three_rop}.
%
\paragraph{Common features}

The role of the $Q$ antennas in compressive radio-interferometry (CRI) is analogous to that of the $Q$ cores in MCFLI. In both cases appears an \emph{interferometric matrix} $\intMOm[]{}$ encoding Fourier samples (or \emph{visibilities}) of a (stationary) 2-D image of interest taken precisely in the difference set $\Omega-\Omega$. The complex exponential terms (resp. $e^{\frac{\im2\pi}{\lambda z}\bs p_q^\top \bs x}$ and $e^{\frac{\im2\pi}{\lambda}\bs p_q^\perp(t)^\top \bs l}$) of the Fourier transforms encoded into the interferometric matrices come both from a dephasing of the electromagnetic signal due to the core/antenna location.  The observed images are both vignetted---$f_{\rm v}:=wf$ and $\vintmap:=g^2\intmap$. In either applications, a \emph{compressive imaging} procedure is considered by applying random ROPs of the interferometric matrix. The \emph{sketching vector} $\bs\alpha$ (and $\bs\beta$ for CRI) is set by choosing the complex amplitude of each core (resp. antenna) thanks to a spatial light modulation.

\paragraph{Differences}

In MCFLI, we image a 2-D plane perpendicular to the distal end of the MCF. In CRI, however, we consider an image in \emph{direction-cosine} coordinates $\bs l$. 
The MCFLI application is completely \emph{stationary}--there is no dependence on time $t$. The measured signal $y$ is deterministic; no expectation needs to be approximated by summing many measurements over time. In CRI, the time dependence of the antenna locations can be exploited to sense many interferometric matrices, thus obtaining a denser Fourier sampling than in MCFLI. This is why we get only $\intMOm[]{}$ in MCFLI, but $\{ \intMOmb[]{} \}\btoB$ in CRI. 
In MCFLI, the symmetric ROPs are \emph{imposed} by the sensing mechanism. In CRI, the ROPs are \emph{pursued} in order to compress the measurement data. Thus, asymmetric ROPs are fully accessible. Furthermore, the noise models are different. In~\eqref{eq:31_mcfli_meas}, the noise $n$ is the thermal noise at the single-pixel detector. In~\eqref{eq:31_three_rop}, the final measurements are obtained by correlations. The thermal noise at the receivers is translated into a deterministic bias $\bs\alpha^* \bs \Sigma_{\bs n} \bs\beta$ that can be removed. The additive Gaussian noise sources in the visibilities comes from the statistical noise induced by the sample covariance and some other model imperfections.

\section{Matrix form of interferometric measurements} \label{app:matrix}

Interferometric measurements associated to a discretized (vignetted) image find a natural matrix formulation. Indeed, writing the discrete image $\intdisc = (\vintmap[\bs n])_{\bs n\in\upto{N_1}^2} \in \Rbb^N$ with $N = N_1^2$ and 
\begin{equation} \label{eq:sigd}
    \vintmap[\bs n]:=\vintmap(\bs n \Delta) = \vintmap(\bs l) \sum_{\bs n \in \upto{N}^2} \delta(\bs l-\bs n\Delta),
\end{equation}
the definition of the interferometric matrix in~\eqref{eq:31_vign} can be particularized to $\intdisc$ and to $t=(b-1/2)MT$ to give
\begin{equation} \label{eq:intMOmbd}
    \big(\intMOmb[]{\intdisc} \big)_{jk} := \Delta^2 \sum_{\bs n\in\upto{N_1}^2} \vintmap[\bs n] e^{\frac{-\im 2\pi}{\lambda} (\bs p_{bk}^\perp-\bs p_{bj}^\perp)^\top \bs n \Delta}.
\end{equation}
In order to express $\intMOmb[]{\intdisc}$ in discrete matrix form, the trick consists in inserting $\sum_{\bs n' \in \upto{N_1}^2} \delta_{\bs n \bs n'}=1$ into~\eqref{eq:intMOmbd} so as to get  
\begin{equation*}
    \big(\intMOmb[]{\intdisc} \big)_{jk} = \Delta^2 \sum_{\bs n, \bs n' \in\upto{N_1}^2} e^{\frac{-\im 2\pi}{\lambda} \bs n^\top \bs p_{bk}^\perp \Delta} \vintmap[\bs n] \delta_{\bs n \bs n'} e^{\frac{\im 2\pi}{\lambda} \bs n^\top \bs p_{bj}^\perp \Delta}.
\end{equation*}
Defining the matrix of complex exponentials $\bs\Gamma_b \in \Cbb^{Q\times N}$ \st $(\bs\Gamma_b)_{qn} := \Delta e^{\frac{-\im2\pi}{\lambda} \bs n^\top \bs p_{bq}^\perp \Delta}$ for the 2-D component $\bs n$ associated to the flattened index $n$ and the diagonal matrix $\bs D_{\intdisc} \in\Rbb^{N\times N}$ filled with the vectorized discrete image $\intdisc$, the interferometric matrix at batch $b$ writes 
\begin{equation} \label{eq:intMOmb1}
    \intMOmb[]{\intdisc} := \bs\Gamma_b \bs D_{\intdisc} \bs\Gamma_b^*.
\end{equation}
While~\eqref{eq:intMOmb1} is already a matrix formulation, it is common to write the decomposition $\bs\Gamma_b := \bs W_b \bs F$ where $\bs F \in \cl H^N$ is the 2-D fast Fourier transform (FFT) matrix and $\bs W_b \in \Cbb^{Q\times N}$ is a matrix interpolating the on-grid frequencies of the FFT to the set of antenna positions $\Omega_b$. One finally gets 
\begin{equation} \label{eq:intMOmb}
    \intMOmb[]{\intdisc} := \bs W_b \bs F \bs D_{\intdisc} \bs F^* \bs W_b^*.
\end{equation}

\section{IROP and MROP as a global ROP model} \label{app:rop}
We explain here that both the IROP and the MROP models can be recast as a ROP model applied to the larger interferometric matrix composed of all batches, \ie 
\begin{equation} \label{eq:32_total_intma}
    \intMOm[]{} := \begin{bmatrix}
        \bs{\cl I}_{\Omega_1} & & \\
        & \ddots & \\
        & & \bs{\cl I}_{\Omega_B}
    \end{bmatrix},
\end{equation}
with specific sketching vectors. Moreover, the reader can easily show that the same conclusions apply to the compressive sensing operators of the IROP and MROP fed by the antenna signals, \ie they are equivalent to the ROPs of a global block diagonal sample covariance matrix built from each covariance $\sampcovmat_b(\cl X_b)$ similarly to~\eqref{eq:32_total_intma}.

First, it is easy to rewrite the $p$-th measurement of the IROP model as
\begin{align} \label{eq:32_sym_irop}
\begin{split}
    \bar z_p &= \bs\alpha_p^* \intMOm[]{} \bs\beta_p = \begin{bmatrix}
        \bs\alpha_{p1}^* & \ldots & \bs\alpha_{pB}^*
    \end{bmatrix}
    \begin{bmatrix}
        \bs{\cl I}_{\Omega_1} & & \\
        & \ddots & \\
        & & \bs{\cl I}_{\Omega_B}
    \end{bmatrix}
    \begin{bmatrix}
        \bs\beta_{p1} \\ \vdots \\ \bs\beta_{pB}
    \end{bmatrix} \\
    &\ts = \sum\btoB \bs\alpha_{pb}^* \bs{\cl I}_{\Omega_b} \bs\beta_{pb}.
\end{split}
\end{align}
Therefore, the ROP measurement vector writes as $\bar{\bs z} = (\bar z_p)_{p=1}^{P} := \ropA (\intMOm[]{})$ where the \emph{sketching operator} $\ropA$ defines $P$ ROP~\cite{chen2015exact,cai15} of any Hermitian matrix $\bs H \in \cl H^{QB}$ with
\begin{equation} \label{eq:32_ROP_op}
    \ropA(\bs H)  := (\fro{\bs \alpha_p \bs \beta_p^*}{\bs H} )_{p=1}^{P}.
\end{equation} 
The name IROP given to this approach appears even more clearly in~\eqref{eq:32_sym_irop}; each ROP measurement is given as a ROP of the total interferometric matrix---integrating all the batches together. This viewpoint is useful for the guarantees given in Sec.~\ref{sec:recovery}.

Second, regarding the MROP model, by introducing the modulation vectors $\bs\epsilon_m,{\bs\epsilon}'_m \in \{\pm 1\}^B$ with $\epsilon_{mb} \simiid \cl U\{\pm1\}$, $\forall m \in \upto{M}$, and the diagonal matrices $\bs D_{\bs\epsilon_m}, \bs D_{{\bs\epsilon}'_m} \in \{ \pm 1, 0 \}^{QB\times QB}$ with $\bs D_{\bs\epsilon_m} := \diag(\bs\epsilon_m) \otimes \Id_Q$,
and equivalently for $\bs D_{{\bs\epsilon}'_m}$, the $(p,m)$-th measurement writes $\bar z_{pm} = (\bs D_{\bs\epsilon_m} \bs\alpha_p)^* \intMOm[]{} (\bs D_{{\bs\epsilon}'_m} \bs\beta_p)$, since $\bar z_{pm}$ equals
\begin{align} \label{eq:32_modul_interpret}
\begin{split}
     &\begin{bmatrix}
        \epsilon_{m1} \bs\alpha_{p1}^* & \ldots & \epsilon_{mB} \bs\alpha_{pB}^*
    \end{bmatrix} 
    \begin{bmatrix}
        \bs{\cl I}_{\Omega_1} & & \\
        & \ddots & \\
        & & \bs{\cl I}_{\Omega_B}
    \end{bmatrix}
    \begin{bmatrix}
        \epsilon'_{m1} \bs\beta_{p1} \\ \vdots \\ \epsilon'_{mB} \bs\beta_{pB}
    \end{bmatrix} \\
    &= \sum\btoB \epsilon_{mb} \epsilon'_{mb} \bs\alpha_{pb}^* \bs{\cl I}_{\Omega_b} \bs\beta_{pb} = \sum\btoB \epsilon_{mb} \epsilon'_{mb} \bar y_{pb} = \sum\btoB \gamma_{mb} \bar y_{pb}.
\end{split}
\end{align}
Setting $\gamma_{mb} := \epsilon_{mb} \epsilon'_{mb}$, we recover in~\eqref{eq:32_modul_interpret} the modulation principle introduced in~\eqref{eq:32_modul1}. 

In summary, the IROP and the MROP models are thus associated with the rank-one projections of the total interferometric matrix $\intMOm[]{}$ over the specific sketching vectors $\{\bs\alpha_{p}, \bs\beta_{p}\}_{p=1}^P$ and $\{ \bs D_{\bs\epsilon_m} \bs\alpha_p, \bs D_{{\bs\epsilon}'_m} \bs\beta_p\}_{p\in \upto{P},m \in \upto{M}}$, respectively. We also observe that for the MROP, the sketching vectors are \emph{structured} and not densely random: while the corresponding sketching vectors collectively occupy a space of $2MP \times QB$ dimensions, they depends only on $2B(PQ + M)$ random variables, \ie the ratio between their free parameters and their total size is $(2MP \times QB) / 2B(PQ + M) = (1/M) + (1/PQ)$. The question of showing how these structured random sketching vectors enables the respect of the $\RIPto$ by the sensing operator $\bs M\bs D\bs G \bs F$ (hence generalizing \ref{prop:24_rip_rop} to the MROP) is left to a future work.

\section{Proof of Prop.~\ref{prop:24_rip_rop}} \label{app:proof}

The proof Prop.~\ref{prop:24_rip_rop} will follow the same reasoning as the proof given in~\cite[App.~C]{leblanc24}, except that the concentration of symmetric ROP measurements around the matrix they are projecting, namely~\cite[Lemma~7]{leblanc24}, is now replaced by a tighter concentration for \emph{asymmetric} ROPs given in~\cite[Prop.~1]{chen2015exact}. 

We recall the ROP interpretation $\bar{\bs z} := \ropA (\intMOm[]{})$ of the IROP model given in App.~\ref{app:rop}. Two key observations can be made about $\intMOm[]{}$. First $\intMOm[]{} \in \cl H^{QB}$ as $\intMOmb[]{} \in \cl H^Q~~ \forall b \in \upto{B}$. Second, 
\begin{equation*} 
    \norm{\intMOm[]{}}{\rm F}^2 = \sum\btoB \norm{\intMOmb[]{}}{\rm F}^2 = \sum\btoB \norm[\big]{\bs G_b \bs F \intdisc }{2}^2 
    = \norm[\big]{\bs G \bs F \intdisc}{2}^2.
\end{equation*}
The same holds for the \emph{hollowed} interferometric matrix $\bc I_\rmh := \intMOm[]{} - \diag(\diag(\intMOm[]{}))$ that results from the centering step described in Sec.~\ref{sec:centering}:
\begin{equation} \label{eq:37_hollow_norms}
    \norm{\bc I_{\rm h}}{\rm F}^2 = \norm[\big]{\bs G_0 \bs F \intdisc}{2}^2,
\end{equation}
with $\bs G_0$ removing the DC component from $\bs G$, and $\bar{\bs z}^{\rm c} := \ropA (\bc I_\rmh)$.
\begin{remark}
    With these two observations, the proof given in~\cite[App.~C]{leblanc24} could be entirely reused, providing a $\RIPto$ for symmetric ROP measurements of zeromean images. However, it provides a less tight concentration.
\end{remark}
Next,~\cite[Prop.~1]{chen2015exact} is reminded here.
\begin{lemma}[Concentration of ROP in the $\ell_1$-norm~\cite{chen2015exact}] \label{lem:27_concentration_ROP_A}
Supposing Assumption~\ref{h:sketch-distrib} holds, given a matrix $\bc J \in \cl H^{QB}$, there exist universal constants $c_1,c_2,c_3>0$ such that with probability exceeding $1 - \exp(- c_3 P)$, 
    \begin{equation} \label{eq:37_concentration_A}
    c_1 \norm{\bc J}{\rm F}  \leq \tinv{P} \|\ropA(\bc J)\|_1 \leq c_2 \norm{\bc J}{\rm F}.
\end{equation}
\end{lemma}
As a simple corollary to the previous lemma, we can now establish the concentration of $\CRIop \intdisc \in {\Rbb}^{P}_{+}$ in the $\ell_1$-norm for an arbitrary $K$-sparse vector $\intdisc \in \Sigma_K$.
\begin{corollary}[Concentration of $\CRIop$ in the $\ell_1$-norm] \label{cor:27_concentration_B}
In the context of Lemma~\ref{lem:27_concentration_ROP_A}, suppose that assumptions \ref{h:bounded-FOV}-\ref{h:sketch-distrib} are respected, with~\ref{h:rip-visibility} set to sparsity $K_0 >0$ and distortion $\delta=1/2$. Given $\intdisc \in \Sigma_K$, and the operator $\CRIop$ defined in~(\ref{eq:32_irop}) from the $P$ ROP measurements and the $V=|\cl V_0|=Q(Q-1)B$ nonzero visibilities with
\begin{equation*}
V \geq 4 K_0\, {\rm plog}(N,K_0),
\end{equation*}
we have, with a failure probability smaller than $\exp(- c_5 P)$, for some $c_5>0$,  
\begin{equation*} \label{eq:37_RIP_concentration-for-ilerop}
    c_1' \norm{\intdisc}{2}  \leq \tinv{P} \norm{\CRIop\intdisc}{1} \leq c_2' \norm{\intdisc}{2}.
\end{equation*} 
\end{corollary}
\begin{proof} 
Given $\intdisc \in \Sigma_{K_0}$ and $\bc J_\rmh$ the hollow version of $\bc J = \bs D_{\bs W} \bs F \bs D_{\intdisc} \bs F^* \bs D_{\bs W}^* \in \cl H^{QB}$ where $\bs D_{\bs W}$ is the block-diagonal matrix filled with the interpolation matrices $\{ \bs W_b \}\btoB$, let us assume that~\eqref{eq:37_concentration_A} holds on $\bc J_\rmh$, an event with probability of failure smaller than $\exp(- c_3 P)$ with $c_3>0$. We first note that $\norm{\bc J_\rmh}{\rm F} = \norm{\bs G_0 \bs F \intdisc }{2}$ from~\eqref{eq:37_hollow_norms}. Second,
\begin{equation} \label{eq:37_tmp-proof-ilerop-concent}
  \tfrac{1}{2} \norm{\intdisc}{2}^2 \leq \tfrac{N}{\varpi^2 V} \norm{\bs G_0 \bs F \intdisc }{2}^2 \leq \tfrac{3}{2} \norm{\intdisc}{2}^2.
\end{equation}
since from Assumption~\ref{h:rip-visibility} the matrix $\bs G_0 \bs F$ respects the $\RIP(\Sigma_{K_0},\delta=1/2)$ as soon as $V = Q(Q-1)B \geq 4 K_0\, {\rm plog}(N,K_0)$.
Therefore, since $\bar{\bs z}^\rmc = \CRIop \intdisc = \ropA(\bc J_{\rm h})$, using~\eqref{eq:37_concentration_A} gives
\begin{align*}
\tinv{P} \norm{\CRIop \intdisc}{1} &\geq c_1 \norm{\bc J_{\rm h}}{\rm F} = c_1  \norm{\bs G_0 \bs F \intdisc }{2} \\
&\geq c_1 \varpi \sqrt{\tfrac{V}{2N}}\, \norm{\intdisc}{2}.
\end{align*}
Similarly, we get
\begin{equation*}
\tinv{P} \norm{\CRIop \intdisc}{1} \leq c_2 \varpi \sqrt{\tfrac{3V}{2N}}\, \norm{\intdisc}{2},
\end{equation*}
concluding the proof with $c_1':= c_1 \varpi \sqrt{\tfrac{V}{2N}}$ and $c_2':= c_2 \varpi \sqrt{\tfrac{3V}{2N}}$.
\end{proof}

We are now ready to prove Prop.~\ref{prop:24_rip_rop}. We will follow the standard proof strategy developed in~\cite{baraniuk08}. By homogeneity of the $\RIPto$, we restrict the proof to unit vectors $\intdisc$ of $\Sigma_\splev$, \ie $\intdisc \in \Sigma_\splev^* := \Sigma_\splev \cap \bb S^{N-1}_2$. 

Given a radius $0<\lambda < 1$, let $\cl G_\lambda \subset \Sigma^*_\splev$ be a $\lambda$ covering of $\Sigma^*_\splev$, \ie for all $\intdisc \in \Sigma^*_\splev$, there exists a $\intdisc' \in \cl G_\lambda$, with $\supp \intdisc' = \supp \intdisc$, such that $\|\intdisc - \intdisc'\| \leq \lambda$. Such a covering exists and its cardinality is smaller than ${N \choose \splev}(1 + \frac{2}{\lambda})^{\splev} \leq (\frac{3eN}{\splev\lambda})^\splev$~\cite{baraniuk08}.

Invoking Cor.~\ref{cor:27_concentration_B}, we can apply the union bound to all points of the covering so that
\begin{equation} \label{eq:37_UB-concent-ileop}
  \ \forall \intdisc' \in \cl G_\lambda,\  c_1' \leq \tinv{P} \norm{\CRIop \intdisc'}{1} \leq c_2',
\end{equation}
holds with failure probability smaller than
\begin{equation*}
    (\tfrac{3eN}{\splev\lambda})^\splev \exp(- c_3 P ) \leq \exp(\splev \ln(\tfrac{3eN}{\splev\lambda}) - c_3 P). 
\end{equation*}
Therefore, there exists a constant $C > 0$ such that, if $P \geq C \splev \ln(\frac{3eN}{\splev\lambda})$, then~\eqref{eq:37_UB-concent-ileop} holds with probability exceeding $1 - \exp(- c P)$, for some $c > 0$.

Let us assume that this event holds. Then, for any $\intdisc \in \Sigma_\splev$, 
\begin{align*}
  \tinv{P} \norm{\CRIop \intdisc}{1} & \leq \tinv{P} \norm{\CRIop \intdisc'}{1} + \tinv{P} \norm{\CRIop(\intdisc - \intdisc')}{1} \\
    &\leq c_2' + \tinv{P} \norm[\big]{ \CRIop\tfrac{\intdisc - \intdisc'}{\norm{\intdisc - \intdisc'}{2} }}{1} \norm{\intdisc - \intdisc'}{2} \\
    &\leq c_2' + \tinv{P} \norm{\CRIop \bs r}{1} \lambda,
\end{align*}
with the unit vector $\bs r := \frac{\intdisc - \intdisc'}{\|\intdisc - \intdisc'\|}$.  However, this vector $\bs r$ is itself $\splev$-sparse since $\intdisc$ and $\intdisc'$ share the same support. Therefore, applying recursively the same argument on the last term above, and using the fact that $\|\CRIop \bs w\|_1$ is bounded for any unit vector $\bs w$, we get
$\tinv{P} \|\CRIop \bs r\|_1\lambda \leq c_2' \sum_{j\geq 1} \lambda^j = \frac{\lambda}{1-\lambda} c_2'$.

Consequently, since we also have 
\begin{align*}
  \tinv{P} \|\CRIop \intdisc\|_1&\ts \geq \tinv{P} \|\CRIop \intdisc'\|_1 - \tinv{P} \|\CRIop(\intdisc - \intdisc')\|_1\\
            &\geq c_1' - \tinv{P} \|\CRIop \bs r \|_1 \lambda,
\end{align*}
we conclude that 
\begin{equation*}
    \frac{1-2\lambda}{1-\lambda} c_1' \leq \tinv{P} \|\CRIop \intdisc\|_1 \leq \frac{1}{1-\lambda} c_2',
\end{equation*}
Picking $\lambda = 1/4$ finally shows that, under the conditions described above, $\CRIop$ respects the $\RIPto(\Sigma_\splev,m_\splev,M_\splev)$ with $m_\splev > \frac{2c_1'}{3}$, and  $M_\splev < \frac{4 c_2'}{3}$.
\end{appendices}

\bibliographystyle{IEEEtran}
\bibliography{biblio}

\begin{IEEEbiography}[{\includegraphics[width=1in,height=1.25in,clip,keepaspectratio]{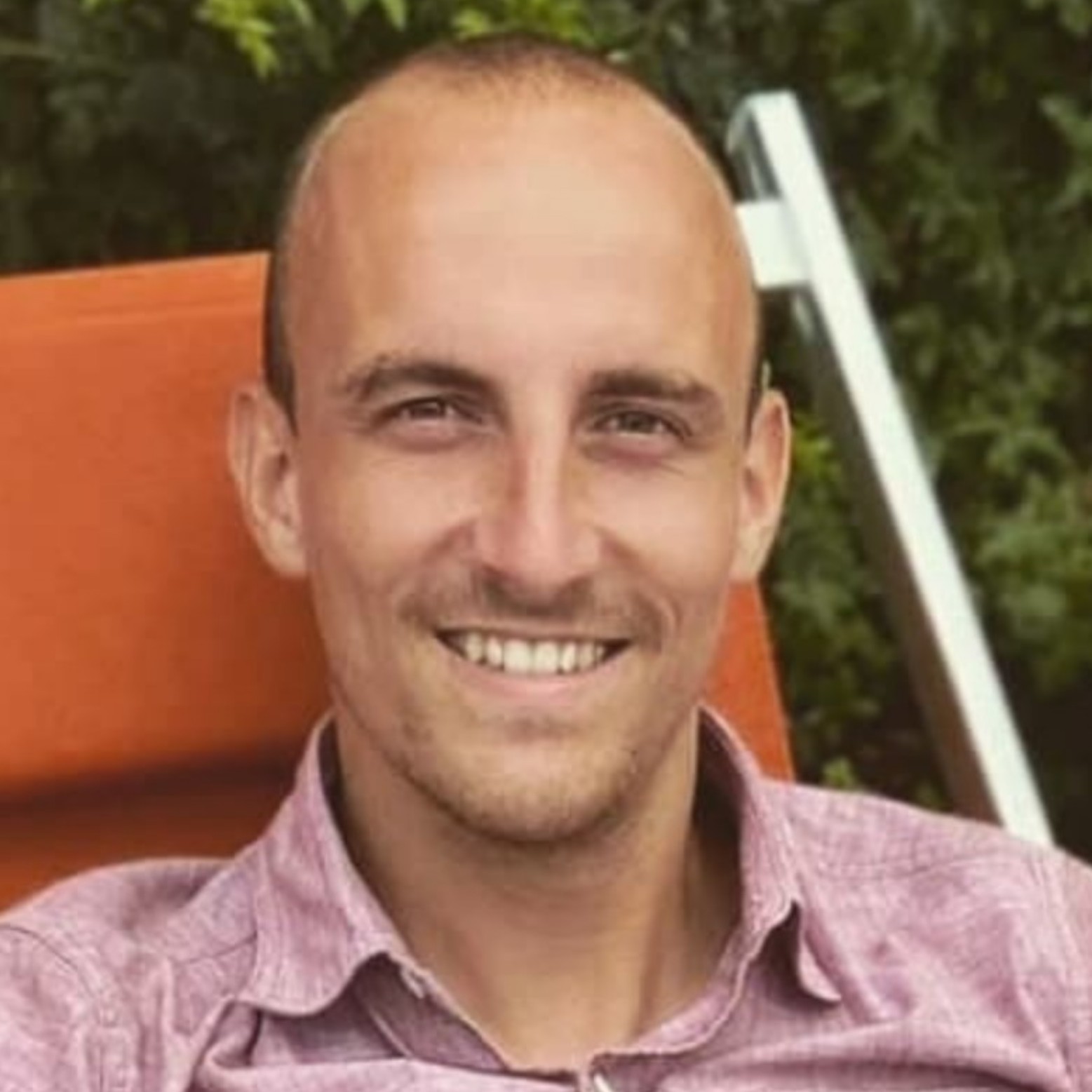}}]{Olivier Leblanc} received the B.Sc., M.Sc., and Ph.D. degrees in electrical engineering from the Mathematical Engineering (INMA) department, ICTEAM/UCLouvain, Louvain-la-Neuve, Belgium, in 2018, 2020, and 2024 respectively. His research interests include computational imaging and compressive sensing (\url{https://olivierleblanc.github.io} for more information).
\end{IEEEbiography}

\begin{IEEEbiography}[{\includegraphics[width=1in,height=1.25in,clip,keepaspectratio]{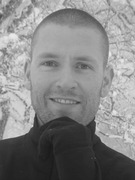}}]{Dr. Yves Wiaux} is a Professor of Imaging Sciences at Heriot-Watt University Edinburgh. He is also an Academic Guest at EPFL and a Honorary Fellow at the University of Edinburgh. He leads the Biomedical and Astronomical Signal Processing Laboratory (BASP; \url{https://basp.site.hw.ac.uk/}), developing cutting-edge research in computational imaging, from theory and algorithms to applications in astronomy and medicine. He is also the chair of the BASP Frontiers Conference series, aiming at promoting methodological synergies between astronomical and medical imaging.
\end{IEEEbiography}

\begin{IEEEbiography}[{\includegraphics[width=1in,height=1.25in,clip,keepaspectratio]{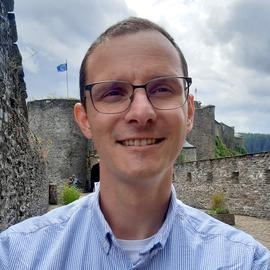}}]{Laurent Jacques} received the B.Sc., M.Sc., and Ph.D. degrees in mathematical physics in 1996, 1998, and 2004, respectively. He has been a FNRS Research Associate from 2012 till 2022 with Image and Signal Processing Group, ICTEAM, UCLouvain, Louvain-La-Neuve, Belgium, and is now professor at the same place. His research interests include sparse signal representations, compressive sensing theory and applications, and computational sensing (\url{https://laurentjacques.gitlab.io} for more information).
\end{IEEEbiography}

\end{document}